\documentclass{article}
\usepackage[english]{babel}\usepackage[english]{babel}

\usepackage[pagewise]{lineno}%\linenumbers
\usepackage[normalem]{ulem}
\usepackage[T1]{fontenc}
\usepackage[utf8]{inputenc}
\usepackage{amsmath,amssymb,mathrsfs}
\usepackage{amsthm}
\usepackage{dsfont}
\usepackage{graphicx}
\usepackage[dvipsnames]{xcolor}
\usepackage{empheq}
\usepackage{faktor}
\usepackage{hyperref}
\usepackage{pgfplots}
\usepgfplotslibrary{fillbetween} 
\usepackage{array}
\usepackage{makecell}
\usepackage{xcolor}
\usepackage{soul}

\definecolor{pastelBlue}     {RGB}{174,198,207}
\definecolor{pastelGreen}    {RGB}{119,221,119}
\definecolor{pastelOrange}   {RGB}{255,179,71}
\definecolor{pastelRed}      {RGB}{255,105,97}
\definecolor{pastelPurple}   {RGB}{179,158,181}
\definecolor{pastelYellow}   {RGB}{253,253,150}
\definecolor{pastelPink}     {RGB}{255,209,220}
\definecolor{pastelTeal}     {RGB}{129,216,208}
\definecolor{pastelGray}     {RGB}{207,207,207}
\usepackage{csquotes}

\usepackage[style=numeric]{biblatex}
\usepackage{subcaption}
\usepackage{adjustbox}

\usepackage{array}
\usepackage{booktabs}
\usepackage{makecell}

\usepackage{tikz}
\usetikzlibrary{shapes,shadows,arrows,graphs}
\usetikzlibrary{decorations.pathreplacing,positioning}

\usepackage{graphics}

\usepackage{vmargin}

\setmarginsrb{2.3cm}{2cm}{2.3cm}{2.1cm}{0cm}{0cm}{0cm}{1.2cm}

\usepackage{comment}

\newcommand{\lb}{\lambda}

\newcommand{\beq}{\begin{equation}}
\newcommand{\eeq}{\end{equation}}
\renewcommand{\epsilon}{\varepsilon}

\renewcommand{\leq}{\leqslant}
\renewcommand{\geq}{\geqslant}
\renewcommand{\d}{\mathrm{d}\, }
\newcommand{\R}{\mathcal{R}}

\newcommand{\RO}{\mathcal{R}_0}
\newcommand{\I}{I_{\text{tot}}}

\renewcommand{\d}{\, \mathrm{d}}

\newtheorem{theorem}{Theorem}
\newtheorem*{theorem*}{Theorem}

\newtheorem{lemma}[theorem]{Lemma}
\newtheorem{coro}[theorem]{Corollary}

\newcommand{\hm}[1]{\textcolor{purple}{#1}}
\newcommand{\gdm}[1]{\textcolor{cyan}{#1}}

\title{Incentives for self-isolation based on incidence rather than
prevalence could help to flatten the curve: a modelling study}

\author{Giulia \textsc{de Meijere} \thanks{Tampere Complexity Lab, Faculty of Information Technology and Communication Sciences, Tampere University, 33720, Tampere, Finland. Email: giulia.demeijere@tuni.fi}, Hugo \textsc{Martin} \thanks{IRMAR, Université de Rennes, CNRS, IRMAR - UMR 6625, 35000 Rennes \textit{and} INRAE, Agrocampus Ouest, Université de Rennes, IGEPP, Le Rheu, France. Now at Insitut Denis Poisson (IDP), Université de Tours, 37000 Tours, France. Email: hugo.martin@univ-tours.fr}}

\date{}

\begin{document}

\maketitle 

\begin{abstract}
In recent years, numerous advances have been made in understanding how epidemic dynamics is affected by changes in individual behaviours. We propose an SIS-based compartmental model to tackle the simultaneous and coupled evolution of an outbreak and of the adoption by individuals of the isolation measure. The compliance with self-isolation is described with the help of the imitation dynamics framework. Individuals are incentivised to isolate based on the prevalence and the incidence rate of the outbreak, and are tempted to defy isolation recommendations depending on the duration of isolation and on the cost of putting social interactions on hold. We are able to derive analytical results on the equilibria of the model under the homogeneous mean-field approximation. Simulating the compartmental model on empirical networks, we also do a preliminary check of
the impact of a network structure on our analytical predictions. We find that the dynamics collapses to surprisingly simple regimes where either the imitation dynamics no longer plays a role or the equilibrium prevalence depends on only two parameters of the model, namely the cost and the relative time spent in isolation. Whether individuals prioritise disease prevalence or incidence as an indicator of the state of the outbreak appears to play no role on the equilibria of the dynamics. However, it turns out that favouring incidence may help to flatten the curve in the transient phase of the dynamics. We also find a fair agreement between our analytical predictions and simulations run on an empirical multiplex network.
\end{abstract}

\noindent\textbf{Keywords:} Compartmental model; Imitation dynamics; Self-isolation; Source of information; Asymptotic behaviour; Transient phase; Minimization of prevalence

\paragraph{Abbreviations}
\begin{itemize}
\item NPI: Non-Pharmaceutical Intervention
\item SIS: Susceptible-Infected-Susceptible
\item IDF: Isolation-Delay-Fatigue
\item CNS: Copenhagen Network Study
\item ODE: Ordinary Differential Equation
\item DFFC: Disease-Free with Full Compliance
\item DFNC: Disease-Free with No Compliance
\item ENC: Endemic with No Compliance
\item EPC: Endemic with Partial Compliance
\item EFC: Endemic with Full Compliance
\item PDE: Partial Differential Equation
\end{itemize}

\section{Introduction}
\par

The survival and well-being of populations can be seriously threatened by the circulation of communicable diseases. At the same time, the very individuals concerned with the spread of a pathogen are often actively involved in the effort to mitigate the spread and to control the pathogen's health-care impact.
Indeed, mitigation measures usually rely on the encouraged but voluntary adoption by individuals of attitudes that are expected to lower the probability of new infections, and thus benefit all. Voluntary adoption as a behavioural decision is however believed to be more than just a spontaneous immutable and homogeneous  state of a population.  Multiple surveys have identified individual factors (e.g. psychological, socio-economic) and message delivery factors (communication about the importance of a measure, and about the severity of the circulating pathogen) that influence whether a person can or will adopt a measure~\cite{Setbon_2010,Raude_2020,Peretti-Watel_2020,Bartolo19}. In addition to these factors, people also form such decisions based on the influence of their social surroundings~\cite{Montgomery_2020,Aghaeeyan_2024, Betsch10}.
\par
Despite a rapidly growing literature on the coupled evolution of disease and adoption of mitigation measures~\cite{Funk_2010_review,Funk_2015,Wang_2015,Verelst_2016,Chang_2020,Bedson_2021,Hamilton_2024,Lejeune_2025,Reitenbach_2025, Wang16, Weston18, Onofriobook, book_Tanimoto}, the complex phenomenology that arises from it remains only partially understood. Various approaches have been proposed to include human behaviour into epidemiological models. Among them, the framework of imitation dynamics relies on the hypothesis that individuals carefully weigh their different options before actively taking side, based on the observed payoffs of others~\cite{Hofbauer_Sigmund}. This framework has been used to model choices related to vaccination~\cite{Bauch_2005,Oraby_2014,Chang_2019,Yin_2022,Khan_2023,Khan_2024}, social distancing~\cite{Poletti_2009,Zhao_2020,Martcheva_2021,Cascante22}, mask-wearing~\cite{Kabir_2021,Traulsen_2023}
and the use of insecticide-treated nets~\cite{Asfaw_2018}. For certain diseases, seasonality plays a pivotal role, and is included into mathematical models. However, even without such periodic modulations, the coupling between the dynamics of the disease and human behaviour is able to generate oscillations as long as the population only partially complies with a mitigation measure in place. If individuals update their opinions sufficiently fast relative to the spreading rate of the disease, the oscillations are either necessarily damped~\cite{Martin25} or possibly sustained~\cite{Bauch_2005}. However, the scalar measure representing the state of the outbreak that should be communicated or the rate of exposure to information that would be beneficial for mitigation is not yet thoroughly understood. In the present work, we propose to address these questions, with the aim of better orienting policy makers and communication during present and future outbreaks.
\par
In the initial stages of a novel outbreak, efficient vaccines and drugs are rarely available and populations are forced to resort to non-pharmaceutical interventions (NPI) in order to slow down the spread of the disease. Depending on the disease~\cite{Peak_2017}, the most efficient NPI to mitigate the spread might vary. Self-isolation of known infected individuals is however a standard NPI that relies on the complete interruption of social interactions over a limited period of time, and was heavily relied upon during the COVID-19 pandemic. Although the measure would ideally only isolate individuals while and if they are infectious, logistical reasons often force healthy individuals or recovered individuals to spend time in complete isolation from the community. Some modelling works have accounted for this type of inefficiency~\cite{Tori_2022, Eksin2017}, while others allowed individuals to enter and break self-isolation at any time~\cite{Zhang_2023}.
\par
Here, we focus on the impact of an imitation dynamics on the efficacy of the self-isolation measure. To do so, we build on a compartmental model introduced in~\cite{deMeijere2021} that accounts explicitly for the following inefficiencies in its implementation: partial adoption, pre- and post-isolation infections.
Partial adoption described in the model allows individuals to decide not to isolate at all, as opposed to several other models for isolation or quarantine where quarantine is not a choice~\cite{Hethcote2002, Chen2019}.
For individuals who comply, pre-isolation infection due to a delayed start of quarantine can be a consequence of unawareness of the infected status (asymptomatic or pre-symptomatic infection, unspecific symptoms, absence of accurate diagnostic tests) or of logistical difficulties with interrupting current levels of activity (e.g., work-related). Post-isolation infection due to overhasty termination of the isolation period can instead be a consequence of `isolation fatigue' \cite{Petherick21} or of isolation recommendations requiring individuals to isolate for a duration that does not always cover the entire infectious period of patients.
\par
In the present work we make two substantial additions to this model. First, we introduce the possibility for individuals exiting quarantine before full recovery to re-enter the community with a reduced infectiousness. This lowered contribution to the force of infection can be attributed either to the adoption of precautionary measures (e.g., social distancing, mask-wearing) or to the pathogen load diminishing as the host immune response progresses~\cite{He_2020,Marc_2021}.
% \par
Second and more importantly, while in \cite{deMeijere2021} the probability of deciding whether to isolate or not is constant, we here let it evolve over time. The decision of individuals is made after carefully pondering the cost of isolating and the responsibility towards the community, based on some information on the state of the outbreak at a given time. In the framework of imitation dynamics, the level of prevalence is the most common information that is used by agents to make their decision. 
However, both prevalence and incidence are considered as the main indicators in epidemiology for the description of the state of an outbreak~\cite{White_2017}.
In the present work, we include both indicators to compare their effect on the efficacy of the isolation measure and the dynamics of the disease, when it is the population that self-regulates the levels of adoption (see~\cite{Bliman22} for another model with a composite source of information).
We expect the dynamic and reactive nature of the self-isolation measure to affect the phenomenology of the coupled behaviour and disease dynamics in a different way compared to measures such as vaccination and mask wearing~\cite{Bauch_2005,Martin25}. While self-isolation has previously been modelled in combination with imitation dynamics (albeit in a different context,~\cite{Amaral_2021}), to our knowledge, it has not yet been studied in the presence of such dual source of information.
\par
We analyse the equilibria of our model and perform a local stability analysis. We find five equilibria and no bi-stable state. Surprisingly, the type of information used (i.e., prevalence or incidence) does not affect the equilibrium values or their stability regions. However, prioritising incidence over prevalence when making decisions may reduce the peak height of the epidemic curve, while contributing marginally to the total number of cases. 
This bias remains neutral under low basic reproduction numbers or limited volatility of opinions. In contrast, if the disease propagates fast and individuals update their strategy at a quick pace, emphasising on incidence-based data can help containing the overall impact of the disease.
\par
In Section \ref{section:model} we introduce our model in its homogeneous mean-field formulation and describe its implementation on a two-layers network. Section \ref{section:results} is devoted to the investigation of the dependence of stationary points of the dynamics and their local stability on the model parameters. There, we also explore the impact of a sparser and more heterogeneous network structure on our analytical conclusions, through simulations on an empirical network. Finally, in Section \ref{section:discussion}, we discuss our results and comment on perspectives for future work.

\section{Modelling framework\label{section:model}}

\subsection{ODE model}

\subsubsection{Disease dynamics}
\begin{figure}
    \centering
    \includegraphics[width=11cm]{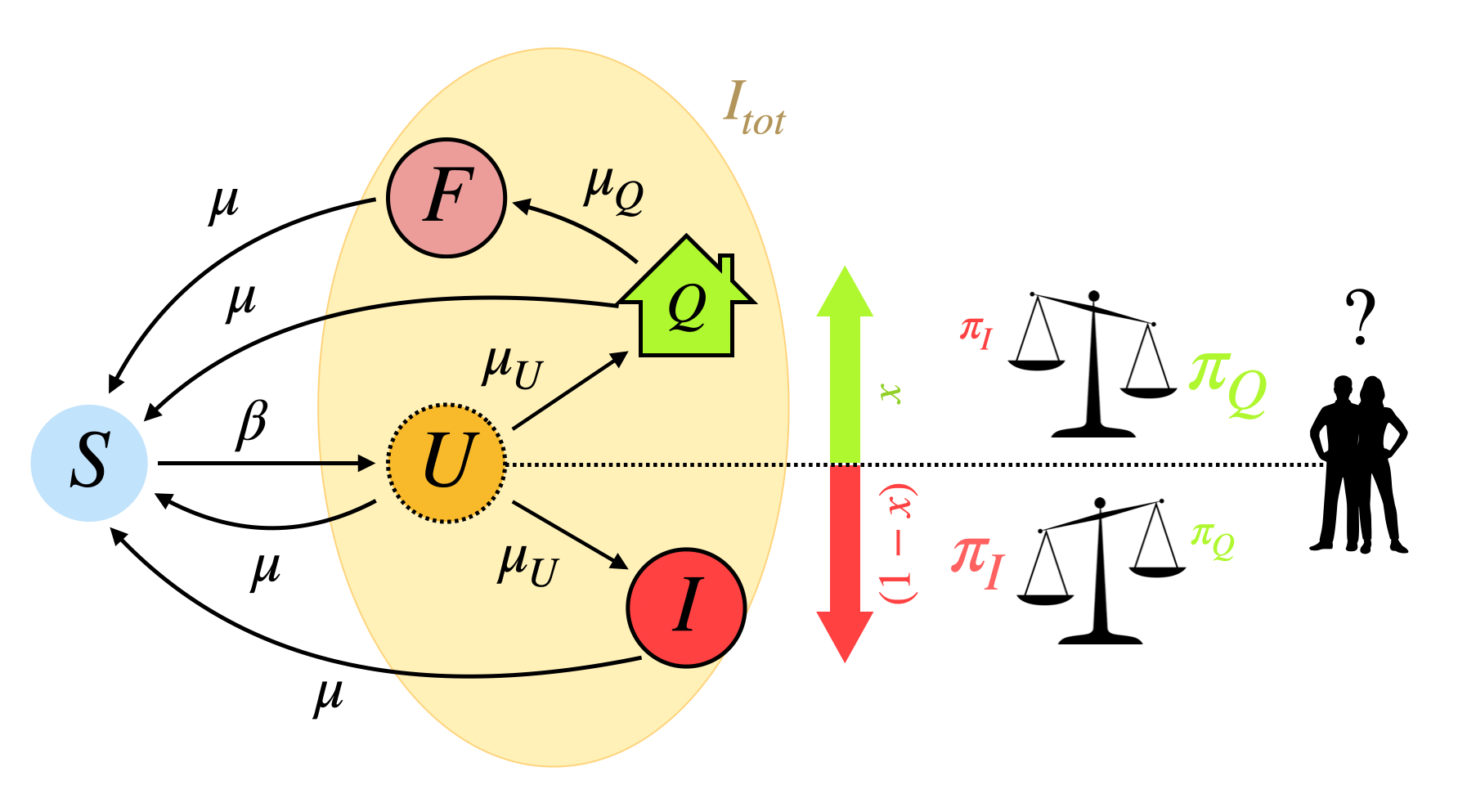}
    \caption{Schematic description of the compartmental model for an SIS-like epidemic evolution with an imitation dynamics affecting the probability that individuals cooperate with the isolation measure.\label{fig:schema}}
\end{figure}

We build a compartmental model based on the standard SIS (Susceptible-Infected-Susceptible) epidemic model \cite{Anderson91}. Let $S$ denote the compartment of susceptible individuals who may contract an infectious disease.
The standard infected compartment of the SIS model is here split into four distinct compartments $U$, $I$, $Q$ and $F$ in order to distinguish contagious from `harmless' (because isolated) infected individuals.
\par
With infection rate $\beta$ and conditional to the contact between a susceptible and a non-isolated infected individual, newly infected individuals enter the compartment $U$ of infected and undecided agents. We assume the standard homogeneous mean-field formulation of the probability of a contact between a susceptible and an infected individual. Individuals in the undecided compartment are infectious and have yet to decide whether to isolate or not, e.g., being unaware of their infected status or for logistical reasons. 
The spontaneous transition out of the undecided compartment occurs with decision rate $\mu_U$. Individuals then either enter the isolated compartment $Q$ with compliance probability $x$, or the still infectious compartment $I$, with the complementary probability $1-x$.
Individuals in the $Q$ compartment are assumed to be perfectly isolated from the community. Although we describe the self-isolation of known infected individuals and not preventive quarantine, we call the isolation compartment $Q$, in order to refer to the already existing literature on SIQS models \cite{Hethcote2002, Young2019, Chen2019, Zhang2017, Esquivel2018, Mancastroppa2020}.
Quarantined individuals can exit the $Q$ compartment to enter yet another infectious compartment $F$ (fatigued) - a transition that occurs spontaneously with rate $\mu_Q$. While individuals in compartment $I$ who refuse to comply with the isolation measure are as contagious as individuals in compartment $U$, fatigued individuals in compartment $F$ are conferred a decreased infectiousness $\beta (1-\epsilon)$, to account for a potential increased awareness after the isolation period (resulting in protective behaviours such as mask wearing and social distancing) or a lower pathogen load due to the natural course of the infection.
Although the rest that individuals get during isolation might help them to recover faster, we neglect it in the present framework and consider that the time to recovery is on average identical regardless of the decision made by individuals, and more in general regardless of the path they undertake in the flow diagramme.
To ensure this, direct transitions from compartments $U$, $Q$, $I$ or $F$ back to the susceptible compartment $S$ occur with the same rate $\mu$ (see Appendix \ref{sec:rate_mu}). Note that undecided individuals may recover before ever making a decision. This happens when the recovery rate is much higher than the decision rate and corresponds to one limit where our model reduces to the standard SIS model.

The fractions of individuals in each compartment overall undergoes the following dynamics:

\[
\begin{cases}
    \dot S &= -\beta S(U + I + (1 - \epsilon) F) + \mu(U + Q + I + F)\\
    \dot U &= \beta S(U + I + (1 - \epsilon) F) - (\mu + \mu_U) U\\
    \dot Q &= \mu_U x U - (\mu + \mu_Q) Q\\
    \dot I &= \mu_U (1 - x) U - \mu I\\
    \dot F &= \mu_Q Q - \mu F,
\end{cases}
\]
with the assumption that demographic changes are negligible, which translates into $S + U + I + Q + F = 1$.
For convenience, let $\I$ denote the total prevalence, i.e. $\I = U + I + Q + F$.

\subsubsection{Coupled disease and opinion dynamics}

\par
To make the here-above model for isolation more realistic, we let the probability $x$ that undecided individuals comply with the isolation measure to be another variable of the model and to evolve in time. We assume that an individual can either be a \textit{cooperator} or a \textit{defector},
depending on whether the individual undergoes self-isolation or not when aware of being infected. The real number $x$ can now be understood as the fraction of cooperators in the population. 
In addition, people are assumed to dynamically change their opinion on self-isolation based on the comparison between the payoff of their own strategy and the payoff of the alternative strategy, in the fashion of the imitation dynamics framework~\cite{Bauch_2005, Ndeffo_Mbah_2012, Wang_2015}. Overall, the variable $x$ undergoes the following dynamics:
\[
\dot x = \sigma x (1 -x) (\pi_Q - \pi_I),
\]
with $\sigma$ the rate at which individuals challenge their opinion, and $\pi_Q$ and $\pi_I$ the payoffs associated to the adoption (cooperation) or not (defection) of the isolation measure, respectively. 
\par
In the present imitation dynamics framework, we are able to account for the deterioration of compliance over longer isolation periods by making the payoff of cooperation decay with the duration of quarantine. More precisely, we model the cost to comply as the product of the daily cost to self-isolate multiplied by the average duration of isolation
\[
\pi_Q = - k \cdot\frac{1}{\mu + \mu_Q}.
\]
The parameter $k$ can be thought of as a straightforward monetary loss from being unable to commute to work, a burden from reorganizing housework, or a social cost from not seeing relatives and friends. A detailed derivation of the average time spent in compartment $Q$ is provided in Appendix~\ref{sec:rate_mu}.
\par
In contrast, the payoff associated with defecting is assumed to decay with the prevalence $\I$ and the incidence $\beta (1 - \I)(\I - Q - \epsilon F)$, with respective weights $k_p$ and $k_i$
\[
\pi_I = - k_p \I - k_i \beta (1 - \I)(\I - Q - \epsilon F).
\]
Combining both payoffs, the time evolution of the proportion of cooperators in the population reads
\[\dot x = \sigma x (1-x) \left(k_p\I + k_i \beta (1 - \I)(\I - Q - \epsilon F) - \frac{k}{\mu + \mu_Q} \right).\]
Note that we assume that the fraction of cooperators is homogeneously distributed across all compartments. In particular, this implies that someone's current opinion is not affected by learning of being infected.

\subsubsection{Final system of equations after nondimensionalization}

Before conducting the mathematical analysis, we nondimensionalise our system relying on the fact that the time of recovery provides a natural timescale $\mu t$ for the system. 
Let $\RO := \frac{\beta}{\mu}$ be the basic reproduction number, also defined as the average number of secondary cases for each primary infected case in an otherwise susceptible population.
After rescaling time by $\mu$, the final system reads
\begin{equation}\label{eq:adim}
\begin{cases}
    \dot{I}_{\text{tot}} &= \RO (1-\I)(\I - Q - \epsilon F) - \I\\
    \dot U &= \RO (1-\I)(\I - Q - \epsilon F) - (1+u) U\\
    \dot Q &= u x U - (1+q) Q\\
    \dot F &= q Q - F\\
    \dot x &= \kappa x (1-x) \left(p\I + (1-p) \RO (1 - \I)(\I - Q - \epsilon F) - \frac{c}{1 + q} \right),
\end{cases}
\end{equation}

where we have introduced
\begin{itemize}
    \item[--] the relative decision-making rate (exit from compartment $U$), denoted $u:= \frac{\mu_U}{\mu}$;
    \item[--] the relative isolation-breaking rate (exit from compartment $Q$), denoted $q:= \frac{\mu_Q}{\mu}$. Note that since the average time spent in isolation is $\frac{1}{\mu+\mu_Q}$, we can say that the average time for recovery is $q + 1$ longer than the average time spent in isolation. Therefore, the mathematical condition $q > 0$ implies an infectious period always longer, on average, than the isolation period;
    \item[--] the normalized opinion update rate $\kappa := \frac{\sigma (k_p + k_i)}{\mu}$ ;
    \item[--] the relative importance given to prevalence when deciding whether to isolate or not, denoted $p := \frac{k_p}{k_p + k_i}$;
    \item[--] the normalized daily cost of self-isolation $c := \frac{k}{\mu(k_p + k_i)}$.
\end{itemize}

\begin{table}
\caption{Summary of the final nondimensionalized model parameters. 
\label{tab:rescaled_param}}
\begin{center}
\begin{adjustbox}{width=\linewidth}
\begin{tabular}{@{}llcc@{}}\toprule
\textbf{Symbol} & \textbf{Name} & \textbf{Plausible range} & \textbf{Reference}
\\\midrule
$\RO \in (0,\infty)$ & Basic reproduction number & $(0.3,6)$ &\cite{vandenDriessche2017, Alimohamadi2020} \\
$\epsilon \in [0,1]$ & Level of protection of infectious fatigued individuals & $(0.05, 0.95)$ &\cite{Leith_2021,Koh_2022}\\
$u \in (0, \infty)$ & Ratio of the decision-making rate over the recovery rate & $(0, 5)$ & \\
$q \in (0, \infty)$ & Ratio of the isolation-breaking rate over the recovery rate & $(0, 2)$ & \\
$\kappa \in (0, \infty)$ & Rescaled rate of opinion update & $(0, 50)$ & \\
$p \in [0, 1]$ & Relative importance given to prevalence in the `Defect' strategy & $[0, 1]$ & \\
$c \in (0, \infty)$ & Rescaled cost to self-isolation & $(0, 3)$ & \\
\bottomrule
\end{tabular}
\end{adjustbox}
\end{center}
\end{table}

When $x(t=0)=0$ or $u \rightarrow 0$, we recover the SIS system of equations. Instead, the case where $\kappa=0$ is identical to the IDF (Isolation-Delay-Fatigue) model introduced in \cite{deMeijere2021}, where the compliance probability $x$ is constant in time. In this paper, we only consider non-negative values for the parameter $\kappa$.
\par
The parameter $p$ of the model accounts for the relative importance of prevalence and incidence in the decision making of individuals.
Fully tailoring compliance with isolation on the epidemic prevalence targets the containment of prevalence-related quantities, such as hospital saturation levels and mortality. On the other hand, fully relying on incidence focuses on the probability of new infections, hence a refusal to isolate if there are either already too many infected individuals (the chances of meeting a susceptible individual whom they might infect are low) or still too few infected individuals (they would be among the few people isolating).
\par
A schematic description of the compartmental model is provided in Fig.~\ref{fig:schema} and the parameters used in the nondimensionalized model are summarized in Table~\ref{tab:rescaled_param}. 

\subsection{Simulations}
\label{section:network_setting}

The mean-field equations for the imitation dynamics can arise from different implementations of the underlying stochastic dynamics. 
We simulate the dynamics using a Gillespie optimised algorithm \cite{Cota2017} and with the help of the Quasistationary State method, proven to be an efficient way of considering only surviving runs in simulations of absorbing phase transitions \cite{QSmet}.
\par
We distinguish two stochastic implementations that we call `one way' and `two ways'.
The `two ways' implementation allows a change in strategy, even if the other strategy has a lower payoff. It sees the change in strategy just like an infection event, whereby an individual `infects' a neighbour and makes it adopt its strategy. 
If the payoff of the explored strategy though is much lower than the one of the original strategy the `infected' individual quickly tends to transition back to its original strategy.
The `one way' implementation instead only allows the transition if the difference of payoffs between the two strategies is advantageous.
In both implementations, the strategy changes are only performed if the strategy of the individuals who change is still carried by some ($>1$) individuals in the population. 
\par 
We compare the performances of the two designed implementations in terms of their error, i.e. deviation from the ODE model, and in terms of their execution time. More precisely, for each implementation, we simulate the time series of $U$, $I$, $Q$, $F$ and $x$ on a complete network of $N = 500$ nodes. Averaging over a given number of runs, we obtain a time series for each compartment, and compute the `combined mean error' as the sum, over all time steps and compartments, of the absolute value of the difference between these time series and ODE predictions. In Figure~\ref{fig:figure5}a, we compare the distribution of $50$ combined mean errors, for each implementation and for different numbers of runs. For both implementations, the error of simulations run on a complete network decays with the number of runs. However, we find that the error of the `two ways' implementation decays more slowly.
This is a consequence of the fact that the `two ways' implementation allows quickly reversed and thus `inefficient' transitions, whereby individuals change their strategy even though the new strategy incurs a lower payoff.
In Figure~\ref{fig:figure5}b, we show the comparison between the distributions of execution times for each implementation. As a consequence of its inefficient transitions, the `two ways' implementation appears to be also consistently slower in execution time compared to the `one way' implementation.
Moreover, these conclusions appear to persist on a network structure where the stochastic implementation no longer perfectly agrees with the ODE predictions.

\begin{figure}
    \centering
    \includegraphics[width=0.6\textwidth]{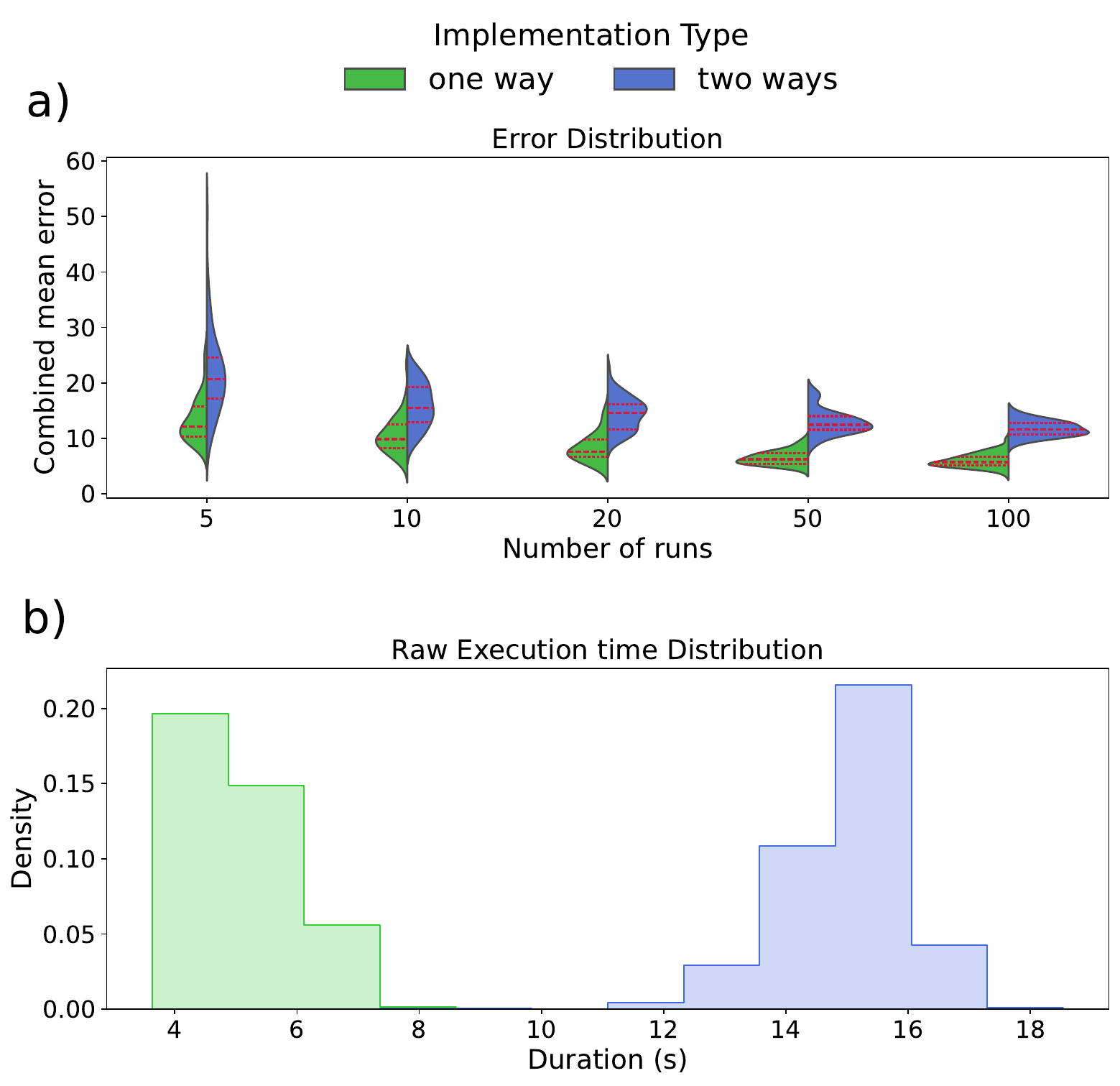}
    \caption{\textbf{The `one way' implementation is more efficient than the `two ways' implementation}. a) Violin plots of the distribution of $50$ combined mean errors for the `one way' (green) and the `two ways' (blue) implementations, and for different numbers of runs. The red lines are the median (dashed) and the lower and upper quartiles (dotted). b) Distribution of the execution time for both the implementations. We consider a complete network of size $N = 500$. \label{fig:figure5}
    } 
\end{figure}

Since the two approaches are both consistent but the `one way' implementation is more efficient, results are here-after only shown for the `one way' implementation.

Our model describes two spreading phenomena occurring on two distinct network layers. We thus consider in principle two different networks, one for the behavioural influence of individuals and one for the physical proximity (see for example~\cite{Qiu2022} for a similar setting). 
For the  empirical network structure, we consider data that was collected as part of the Copenhagen network study (CNS) \cite{Sapiezynski19, Netzschleuder_repo}, a study involving more than 700 students from a university in Copenhagen. Multi-layer temporal data was harvested over a period of four weeks through the mobile phones of the participants. Among other signals, the physical proximity of the students was estimated via Bluetooth signal strengths, and information about their Facebook friendships was also recorded. 
For the propagation of the synthetic disease we use the $\texttt{copenhagen/bt}$ Bluetooth dataset. Instead, for the layer where the adoption of measures propagates, we consider the $\texttt{copenhagen/facebook}$ Facebook network.
We assume that if the number of nodes differs from one layer to the other, it is a consequence of missing data and not a consequence of the inactivity of some nodes in either of the layers. The multi-layer network is aggregated over the entire observation period, and interactions between individuals are assumed to be unweighted and undirected.

\section{Results}
\label{section:results}
A stability analysis of the nondimensionalized ODE model (System~\eqref{eq:adim}) reveals that the system has two disease-free equilibria as well as three endemic equilibria, distinguishable by the fraction of cooperators in the population. The equilibria of the dynamics and their conditions of stability shed light on the role played by pre-isolation and post-isolation infections, partial compliance, increased awareness after isolation, source of information and volatility of opinion. Moreover, simulating the stochastic dynamics of the present compartmental model on a network structure suggests a fair agreement with analytical results although oscillations found in the ODE model may be damped faster on the network.

\paragraph{ODE model: equilibria and stability.}

The technical details of the assessment of the local asymptotic stability are postponed to Appendix~\ref{apdx:stab}. In a nutshell, the procedure involves (i) identifying the equilibria of the model, (ii) computing the Jacobian matrix associated with System~\eqref{eq:adim}, (iii) plugging the equilibria into the Jacobian matrix and (iv) finding the eigenvalues. 
\\
Two of the equilibria correspond to disease-free situations, i.e. characterised by $\I^* = 0$ and thus also $U^* = Q^* = F^* = 0$. Consequently, the equilibria differ only by the level of compliance, which can only take two values. On the one hand, if $x^* = 1$, corresponding to a scenario with full compliance, we call the equilibrium Disease-Free with Full Compliance (DFFC). This equilibrium is unconditionally unstable. On the other hand, if $x^* = 0$, the equilibrium is called Disease-Free with No Compliance (DFNC), and is locally asymptotically stable if and only if $\RO \leq 1$. 
Indeed, in the absence of a threat, the incentive to self-isolate vanishes, leading to the complete absence of cooperators in the population.
\par
Since the other equilibria correspond to endemic situations, we assume $\RO > 1$ from this point onward.
We are able to derive the equilibrium values of the fractions of individuals in compartments $U$, $Q$, $I$ and $F$ along with the share $x$ of cooperators in the population.
In the stationary limit, the fractions of individuals in compartments $U$, $Q$ and $F$ can be expressed as functions of $x^*$ and $\I^*$:
\begin{equation}\label{eq:UQF}
\begin{cases}
    U^* &= \frac{1}{1 + u}\I^*\\
    Q^* &= \left(1 - \frac{1}{1 + u}\right)\frac{1}{1 + q}x^* \I^* = \frac{u}{1+u}\frac{1}{1 + q}x^* \I^*\\
    F^* &= \left(1 - \frac{1}{1 + u}\right)\left(1 - \frac{1}{1 + q}\right)x^*\I^* = \frac{u}{1+u}\frac{q}{1 + q}x^* \I^*,
\end{cases}
\end{equation}
where the stationary values $x^*$ and $\I^*$ solve the following two equations
\begin{equation}\label{eq:Ix}
\begin{cases}
\frac{1}{\RO} & = (1 - \I^*)\left(1 - \frac{u}{1 + u}\frac{1 + \epsilon q}{1 + q}x^*\right)\\
0 & =  x^*(1 - x^*)\left(\I^* - \frac{c}{1 + q}\right),
\end{cases}
\end{equation}
under the assumption that $\kappa > 0$. We are left with the computation of the equilibrium values of $x^*$ and $\I^*$. 

Provided $x^*$ lies in $[0,1]$, the first line of System~\eqref{eq:Ix} gives
\[
    \I^*(x^*) := 1 - \frac{1}{\mathcal R_{x^*}}.
\]
where $\mathcal R_{x^*}$ reads
\begin{equation}\label{eq:Rx}
    \mathcal R_{x^*}:= \left(1-\frac{u}{1+u}\frac{1 + \epsilon q}{1 + q}x^*\right)\RO.
\end{equation}
In the case of no compliance, i.e. $x^* = 0$ (`No Compliance', ENC), this expression reduces to the standard SIS result
\[
\I^* = 1 - \frac{1}{\RO}.
\]
If full compliance is reached, i.e. $x^* = 1$ (`Full Compliance', EFC), the equilibrium prevalence is
\[
\I^* = 1 - \frac{1}{\mathcal{R}_1}.
\]
The parameter $\mathcal{R}_1$ can be understood as the basic reproduction number of a SIQS model with an extra compartment for fatigued individuals (pre- and post-isolation infections are modelled but partial adoption is not).
\par
Finally, in case $0 < x^* <1$ (`Partial Compliance', EPC), the second equation of System~\eqref{eq:Ix} yields a simple dependence of the equilibrium prevalence on the model parameters
\begin{equation}\label{eq:ItotEPC}
\I^* = \frac{c}{1+q},
\end{equation}
and the first equation provides the following equilibrium share of compliers
\begin{equation}\label{eq:xstar}
   x^* := \frac{1 + u}{u}\frac{1 + q}{1 + \epsilon q}\left(1 - \frac{1}{\RO\left(1 - \frac{c}{1+q}\right)}\right).
\end{equation}
Equating this expression to $0$ and to $1$ provides the stability conditions for the ENC and EFC equilibria, as well as the equations for the boundaries that separate the stability regions.
It is worth to note that none of the equations depends on the initial conditions (except the conditions $\kappa > 0$ and $x(0) \in (0,1)$).

All the results regarding the large time asymptotics of the model are summarised in Table~\ref{table:eq}.

\paragraph{The equilibrium prevalence reduces to a surprisingly simple dependence on the model parameters.}

We refer the reader to Figure~\ref{fig:figure2} for heatmaps of the equilibrium prevalence across a variety of parameter combinations. Both the parameters $p$ of the source of information (prevalence vs incidence) and $\kappa$ of the volatility of opinions are omitted as they do not affect the equilibria. Indeed, at equilibrium, the incidence $\RO (1-\I)(\I - Q - \epsilon F)$ equals the prevalence $\I$, and the two contributions of the source of information cancel each other out. 
\par
Note that the intertwined evolution of disease and compliance with the self-isolation measure results in a surprisingly simple set of equilibria. In two regimes -~`Full Compliance' and `No Compliance'~- the dynamics collapses to one where the imitation does not play any role. Instead, in the third and last regime -~`Partial Compliance'~- the equilibrium prevalence reduces to a very simple function of only two model parameters: the (rescaled) cost of isolation $c$ and the relative isolation-breaking rate $q$.
\par
Surprisingly, in this regime, the prevalence does not depend on the delay to enter isolation (through $u$) or on the level of caution of individuals at the exit from isolation (through $\epsilon$). Delay and level of caution rather affect the equilibrium level of compliance $x^*$ and the stability conditions of this regime. In `Partial Compliance', being more careful 
to start promptly the isolation period (a larger $u$) or to limit risky behaviours after the isolation period (a larger $\epsilon$) thus does not restrain the outbreak size, but makes the containment more efficient, allowing the same equilibrium prevalence $\I^*$ with a lower overall compliance of the population $x^*$. 
\par
In the absence of changes in the adoption of the isolation measure, these two parameters are expected instead always to affect the equilibrium prevalence, unless there is no compliance at all.
In Appendix~\ref{section_imitation} we compare analytical predictions in the presence and in the absence of changes in compliance. 
% modification of who complies or not
We show that under the trivial condition that the static level of adherence of the population is given by the equilibrium level of adherence, the equilibrium prevalences are identical. Moreover, depending on the comparison between the initial level of adherence $x(t=0)$ and the equilibrium level of adherence $x^*$, we can predict whether a dynamics with higher rates of imitation will generate consistently higher or lower prevalence in the transient phase.

\paragraph{Optimal isolation duration.}
We find that an optimal isolation period emerges. We typically observe that the equilibrium prevalence is minimized by an intermediate isolation duration, balancing the trade-off between the increased efficacy of extended isolation periods against the importance of widespread adoption for breaking transmission chains.
\par
For a sufficiently large $c$ and sufficiently low $q$ (long isolation duration), the system is in the `No Compliance' regime and the prevalence is independent of $q$.
When $q$ reaches the value at which it leaves the `No Compliance' regime to enter the `Partial Compliance' regime, the prevalence starts decaying monotonically with increasing $q$: shortening the isolation duration allows to increase the population's adoption levels and to better contain the spread of the pathogen. 
Above a certain value of $q$ though, the system changes regime once again and enters the `Full Compliance' regime, where further increasing $q$ (and shortening the isolation duration) makes the prevalence grow again: the adoption level has reached a ceiling and shorter isolation periods can only be detrimental for the containment of disease spread. 
The value of $q$ that minimises the equilibrium prevalence is thus the value of $q$ that separates the `Partial Compliance' and the `Full Compliance' regimes. This value of $q$ corresponds to the longest isolation period such that everyone complies. 
Still, a lower prevalence is always achieved for longer optimal isolation periods.
The improvement offered by choosing the optimal duration might be modest (Figure~\ref{fig:figure2}, bottom panels) or dramatic (top right panel, for rather low costs $c$), depending on the celerity to enter isolation $u$ and on the caution after exiting isolation $\epsilon$.
\par
However, such an optimal isolation period does not exist in all regions of the space of parameters.
In particular, it is necessary that the cost of quarantine is sufficiently large. Indeed, if $c$ is too small ($c < 1 - \frac{1+u}{\mathcal{R}_0}$), a large fraction of the population cooperates, and there is no further compliance to be won by shortening the isolation period.
\par
In Figure~\ref{fig:figure3} we show the impact of different choices of the cost $c$ of quarantine on the existence of an optimal isolation duration. In Fig.~4a, we focus on the central panel of Figure~\ref{fig:figure2}, and propose the exploration of several trajectories in the $c$ vs $q$ space. First, we consider different constant costs $c$ and make the parameter $q$ vary. Then, we postulate various correlations between $c$ and $q$, and explore their impact on the equilibrium prevalence. Since we expect that the perspective of a long isolation duration weighs heavily on complying individuals, we propose two minimal ways in which the cost of isolation decays with $q$ (light green and blue lines). However, we can also imagine that the perceived cost of isolation is the largest for very short isolation periods. We thus also provide one example of the behaviour of the system under such an unrealistic scenario (red line). In Fig.4b, we see how the choice of the cost affects the existence of an optimal isolation period as well as the range of values of prevalence that can be achieved.

\begin{figure}
    \centering
    \includegraphics[width=\textwidth]{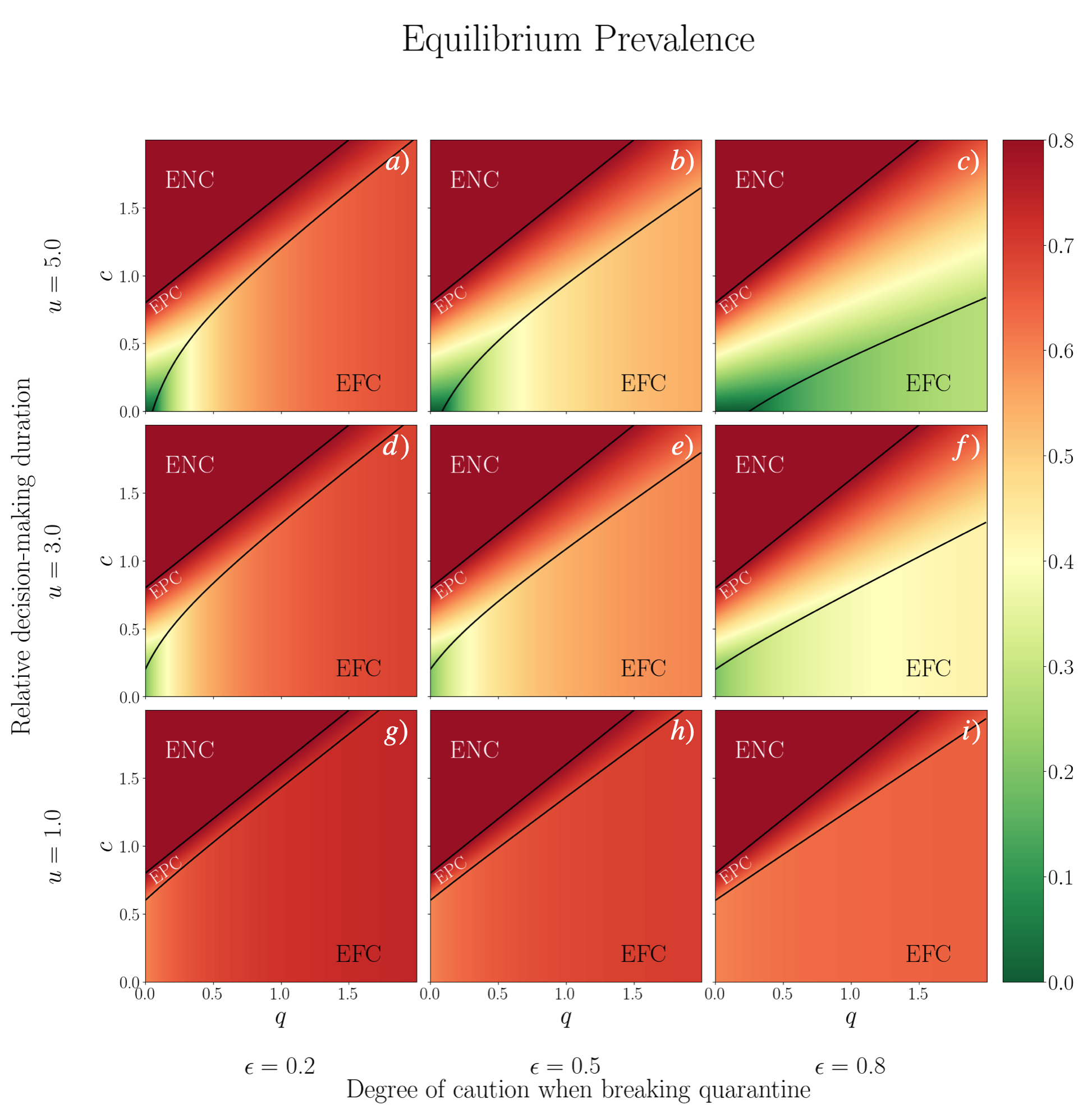}
    \caption{\textbf{Equilbrium prevalence}. Heatmaps of the equilibrium prevalence depending on the cost $c$ of quarantine, the lack of coverage  $q$ of the infectious period by quarantine, the rapidity of entrance in isolation $u$ and the degree of caution $\epsilon$ when breaking quarantine. The limits separating the ENC and EPC regimes, as well as the ones separating the EPC and EFC regimes are shown (black lines). The basic reproduction number is set to $\RO = 5$.}
    \label{fig:figure2}
\end{figure}

\begin{table}
\caption{\textbf{Properties of the equilibria associated to System~\eqref{eq:adim}.} Only the values of $\I^*$ and $x^*$ are displayed, since the other fractions can readily be deduced from their values, using System~\eqref{eq:UQF}. `Existence conditions' are obtained by requiring $\I^*,x^*\in [0,1] \times [0,1]$.\label{table:eq}}
\begin{center}
\begin{adjustbox}{width=\linewidth}
\begin{tabular}{|c|c|c|c|c|c|}%\label{tab:equilibria}
  \hline
  \textbf{Name} & \textbf{Abbr.} & \textbf{Value of $x^*$} & \textbf{Value of $\I^*$} & \textbf{Existence} & \textbf{Stability} \\
  \hline \hline
 \makecell{Disease-Free with\\No Compliance} & DFNC & $0$ & $0$ & $\RO \leq 1$ & $\RO \leq 1$ \\
  \hline
  \makecell{Disease-Free with\\Full Compliance} & DFFC & $1$ & $0$ & $\R_1 \leq 1$ & Never \\
  \hline \hline
  \makecell{Endemic with\\ No Compliance} & ENC & $0$ & $1-\frac{1}{\RO}$ & $\RO > 1$ & $1 - \frac{1}{\RO} \leq \frac{c}{1+q}$ \\
  \hline
  \makecell{Endemic with\\ Partial Compliance} & EPC & \makecell{$x^* \in (0,1)$\\given by~\eqref{eq:xstar}} & $\frac{c}{1+q}$ & \makecell{$\RO > 1$, $\frac{c}{1+q}\in$ \\ $\left(1-\frac{1}{\R_1}, 1 - \frac{1}{\RO}\right)$} & \makecell{At least for $\kappa$\\`small enough'} \\
  \hline
  \makecell{Endemic with\\ Full Compliance} & EFC & $1$ & $1 - \frac{1}{\R_1}$ & $\R_1 > 1$ & $\frac{c}{1+q} \leq 1-\frac{1}{\R_1}$\\
  \hline %\left(1-\right)
\end{tabular}
\end{adjustbox}
\end{center}
\end{table}
% \hm{\Large Ajouter une sous-section ici ? Ce n'est plus uniquement l'EDO}
\paragraph{The source of information plays a role in the transient phase of the dynamics.}
Although the relative importance given to prevalence compared to incidence plays no role on the equilibria, we find that it affects the initial transient phase of the dynamics. More specifically, it appears that the peak of the prevalence is always minimised by a smaller value of $p$, i.e. a preference given to the incidence rate as a source of information as opposed to the prevalence. In Figure~\ref{fig:figure4} we show how increasing the parameter $p$ in the `Partial Compliance' regime leads to a higher peak prevalence without significantly affecting the average prevalence over the entire time window considered, i.e. the temporal average. Similarly, the authors of~\cite{Amaral_2021} found that a parameter of perceived cost could have a large effect on the height of a single wave epidemic peak, while affecting the final epidemic size in a negligible way. In our model it appears that while under certain conditions it is equivalent to use prevalence or incidence as a source of information, in others it is beneficial to consider incidence. This effect is the strongest for faster dynamics, i.e. higher basic reproduction numbers and higher volatility of opinions. This result is surprising in the sense that using incidence as the indicator of the responsibility of individuals towards the community allows them in principle to refuse to isolate when the prevalence is excessively large. However, the incidence grows faster at the beginning of an outbreak compared to the prevalence, allowing a swifter reaction of the population to the outbreak. We stress though that the difference between the curves, in most of the cases, remains small. We also show that these conclusions hold when simulating the model on the multi-layer empirical network described in Section \ref{section:network_setting}.

\paragraph{The decision process only affects undecided individuals.}
Our model demonstrates the capacity for sustained oscillations under specific conditions. 

The coupling of disease and behaviour dynamics in models has been shown to produce oscillations, particularly in the presence of an imitation dynamics and a sufficiently high opinion update rate~\cite{Bauch_2005, Zhang_2023}. However, there are nuances in the nature of the oscillations that can arise. For instance, a model for vaccination~\cite{Bauch_2005} identifies a threshold above which oscillations are sustained, a finding that contrasts with a model for mask-wearing~\cite{Martin25}, where only damped oscillations are observed. This discrepancy can be attributed to the delayed effect of a change in strategy on the disease dynamics.
In~\cite{Martin25}, 
when individuals change strategy, their new behaviour has an immediate impact on the transmission rate. Conversely, in~\cite{Bauch_2005}, a parent's decision to vaccinate their child is irreversible once immunization has occurred at birth. Thus, a change in their opinion can impact the system only through the birth of another child.
\par
Our current model adopts a similar philosophical stance regarding decision irreversibility. When individuals enter the $U$ compartment, they make a definitive choice regarding self-isolation. Any potential regret concerning this strategy necessitates recovery from the current infection, followed by a subsequent reinfection, before a new decision can be made.

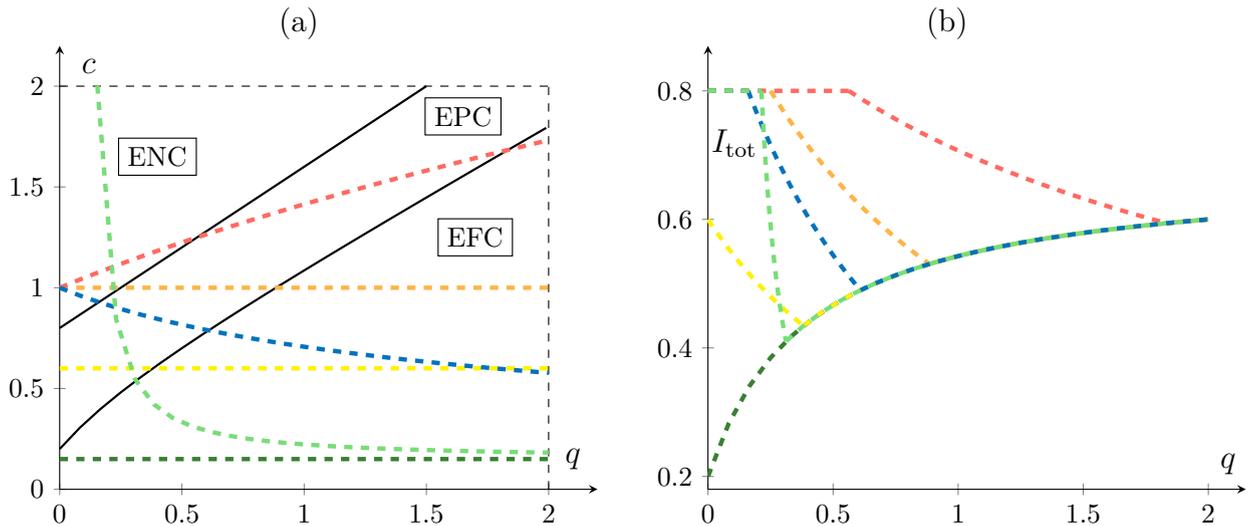
\begin{figure}
  \centering
  \begin{subfigure}[b]{0.48\textwidth}
    \caption{}
    \begin{adjustbox}{width=\linewidth}
      \begin{tikzpicture}
        \begin{axis}[axis lines=left, xmin=0,xmax=2.2, ymin=0,ymax=2.2]
          \addplot[domain=0:2,dashed] coordinates {(2,0)(2,2)};
          \addplot[domain=0:2,dashed, name path=TOP] coordinates {(0,2)(2,2)};
          \node [fill=white,draw=black,anchor=center] at (170,125) {EFC};
          \node [fill=white,draw=black,anchor=center] at (165,185) {EPC};
          \node [fill=white,draw=black,anchor=center] at (40,165) {ENC};

          \addplot[domain=0:1.5, thick, name path=ENC] {(1 - 1 / 5) * (1 + x)};
          \addplot[domain=0:1.99, thick, name path=EFC] {(1 - 1 / ((1 - (3 / (1 + 3)) * ((1 + 0.5 * x) / (1 + x))) * 5)) * (1 + x)};
          \path [name path=BOTTOM] (0,0) -- (200,0); 

          \addplot[domain=0:2, ultra thick, pastelOrange,dashed]{1};
          \addplot[domain=0:2, ultra thick, yellow,dashed]{0.6};
          \addplot[domain=0:2, ultra thick, OliveGreen,dashed]{0.15};
          \addplot[domain=0:2, ultra thick, pastelRed,dashed]{sqrt(1+x)};
          \addplot[domain=0:2, ultra thick, RoyalBlue,dashed]{1/sqrt(1+x)};
          \addplot[domain=0.154:2, ultra thick, pastelGreen,dashed]{0.1*exp(0.4*(1+1/x))};

          \node [anchor=west] at (5,210) {\large {$c$}};
          \node [anchor=north] at (210,25) {\large {$q$}};
        \end{axis}
      \end{tikzpicture}
    \end{adjustbox}
  \end{subfigure}
  \hfill
  % RIGHT PANEL
  \begin{subfigure}[b]{0.48\textwidth}
    \caption{}
    \begin{adjustbox}{width=\linewidth}
      \begin{tikzpicture}
        \begin{axis}[axis lines=left,
          xmin=0,
          xmax=2.15, ymin=0.18,
          ymax=0.87]
          \addplot[domain=0:2, ultra thick, OliveGreen,dashed]{1 - 1/(5*(1-3/(1+3)*(1+0.5*x)/(1+x)))};
          \addplot[domain=0:0.3814871396610922, ultra thick, yellow,dashed]{0.6/(1+x)};
          \addplot[domain=0.3814871396610922:2, ultra thick, yellow,dashed]{1 - 1/(5*(1-3/(1+3)*(1+0.5*x)/(1+x)))};

          \addplot[domain=0:1/(1-1/5) - 1, ultra thick, pastelOrange,dashed]{1 - 1/(5)};
          \addplot[domain=1/(1-1/5) - 1:0.8847996811054291, ultra thick, pastelOrange,dashed]{1/(1+x)};
          \addplot[domain=0.8847996811054291:2, ultra thick, pastelOrange,dashed]{1 - 1/(5*(1-3/(1+3)*(1+0.5*x)/(1+x)))};

          \addplot[domain=0:(3/4)^2, ultra thick, pastelRed,dashed]{1 - 1/(5)};
          \addplot[domain=(3/4)^2:1.834, ultra thick, pastelRed,dashed]{1/sqrt(1+x)};
          \addplot[domain=1.834:2, ultra thick, pastelRed,dashed]{1 - 1/(5*(1-3/(1+3)*(1+0.5*x)/(1+x)))};

          \addplot[domain=0:(5/4)^(2/3) - 1, ultra thick, RoyalBlue,dashed]{1 - 1/(5)};
          \addplot[domain=(5/4)^(2/3) - 1:0.6094, ultra thick, RoyalBlue,dashed]{1/(1+x)^(3/2)};
          \addplot[domain=0.6094:2, ultra thick, RoyalBlue,dashed]{1 - 1/(5*(1-3/(1+3)*(1+0.5*x)/(1+x)))};

          \addplot[domain=0:0.2136, ultra thick, pastelGreen,dashed]{1 - 1/(5)};
          \addplot[domain=0.2136:0.3118, ultra thick, pastelGreen,dashed]{0.1*exp(0.4*(1+1/x))/(1+x)};
          \addplot[domain=0.3118:2, ultra thick, pastelGreen,dashed]{1 - 1/(5*(1-3/(1+3)*(1+0.5*x)/(1+x)))};

          \node [anchor=east] at (23.5,538) {\large {$\I$}};
          \node [anchor=west] at (201,40) {{\large $q$}};
        \end{axis}
      \end{tikzpicture}
    \end{adjustbox}
  \end{subfigure}

  \caption{\textbf{Conditions for an optimal duration of quarantine}. a): Different choices of the parameter $c$: $c = 0.15$ (dark green), $c(q)=0.1*\exp\left(0.4*\left(1 + \frac{1}{q}\right)\right)$ (light green), $c = 0.6$ (yellow), $c(q)=\sqrt{1+q}$ (blue), $c = 1$ (orange) and $c(q)=\frac{1}{\sqrt{1+q}}$ (red). b): Equilibrium prevalence as a function of $q$ when considering the values of $c$ shown in panel a. Parameter values: $\RO = 5$, $u = 3$, $\epsilon = 0.5$. \label{fig:figure3}}
\end{figure}

\begin{figure}
    \centering
    \includegraphics[width=\textwidth]{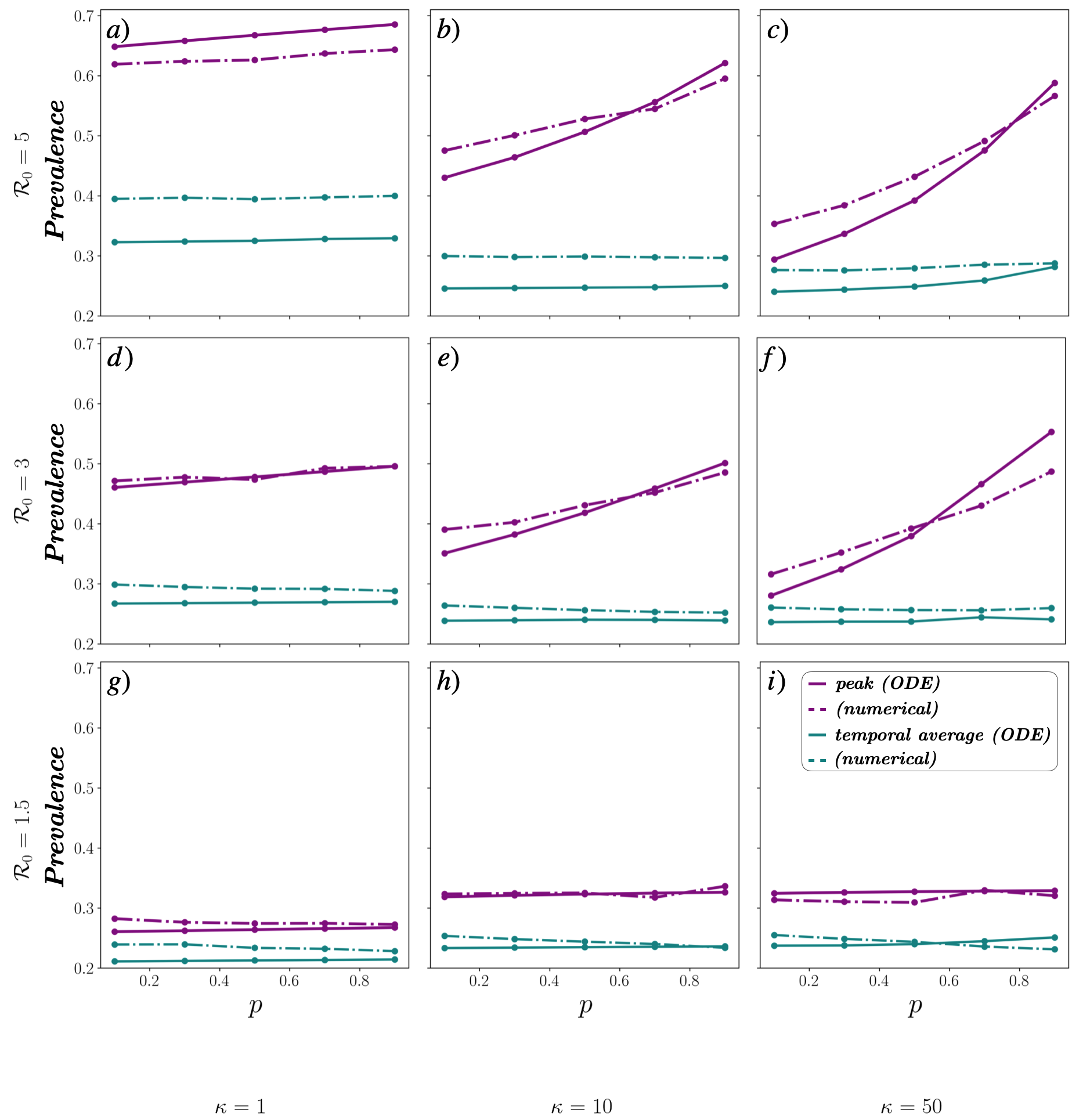}
    \caption{\textbf{The incidence as an indicator of the state of the outbreak can help to flatten the curve}. We show the value of the peak prevalence (purple) and its temporal average (dark cyan) as functions of the parameter $p$, for various choices of the basic reproduction number $\RO$ and of the volatility of opinions $\kappa$. In all panels the system is in the `Partial Compliance' regime. We show both the effect of $p$ in the ODE formulation (continuous line) and in the stochastic simulations run over an empirical two-layers network (dot-dashed line). Parameter values: $u = 5$, $q = 0.25$, $c = 0.3$, $p = 0.4$, $\epsilon = 0.5$, $x_0 = 0.4$.\label{fig:figure4}
    }
    
\end{figure}

\section{Conclusions}
\label{section:discussion}
\par
We build an SIS-based compartmental model that accounts for the simultaneous spread of i) a pathogen on a network of physical contacts and of ii) opinions driving adoption of self-isolation on a network of social influence. By taking into account the coupled evolution of opinions and disease, we want to identify optimal implementations of the self-isolation measure combined with post-quarantine awareness.
\par
The compartmental model that we propose allows for partial adoption of the self-isolation measure, pre-isolation and post-isolation infections. In the compartment of infectious individuals who break quarantine, the infectiousness is assumed to be reduced as a consequence of increased awareness, resulting for example in behaviours such as mask wearing or social distancing. 
The partial adoption of the measure refers to the binary choice that individuals make: to isolate or not to isolate. This decision is governed by an imitation dynamics and depends on a careful comparison of one's current decision with the one of the contacts. The costs of the two possible strategies are carefully weighted by the deciding individuals who are considered to behave rationally: a change of strategy only takes place if the opposite strategy would incur a lower cost. On the one hand, the cost of complying with the isolation measure is assumed to increase with the isolation duration and with the incurred cost. On the other hand, the cost of not complying arises from the state of the epidemic. The prevalence is often considered as the indicator of the state of an outbreak affecting the compliance of players in imitation dynamics models. 
However, the incidence rate is another standard indicator in epidemiology, and we here study how this additional source of information can affect the dynamics of the coupled disease-behaviour model. Instead of using solely the prevalence to assess the state of the epidemics, we thus model the cost of quarantine through a convex combination of prevalence and incidence.

\vspace{1em}

We find that despite a seemingly complex intertwined behaviour, the equilibria of the dynamics are fairly simple. 
There are two disease-free equilibria and three endemic equilibria that can be distinguished by the share of complying individuals in the stationary limit. The `No Compliance' and the `Full Compliance' endemic equilibria are stable when the cost of quarantine is too big or too small, respectively. In the regions of the parameter space where these equilibria are stable, the system collapses to a long-time state that is independent of the imitation dynamics. The third endemic equilibrium is characterised by a partial level of compliance in the population. This regime becomes particularly relevant in more optimistic scenarios, characterised by lower basic reproduction numbers, shorter delays to enter isolation (large $u$) and larger levels of caution after the isolation period (large $\epsilon$). 
In this regime, the equilibrium prevalence reduces to the total cost to self-isolate which is the product of the (relative) average time spent in isolation (through $q$) and the cost per time unit (through $c$) (see Equation~\eqref{eq:ItotEPC}). In contrast, the equilibrium fraction of complying individuals decays with $u$ and $\epsilon$. Consequently, encouraging individuals to hurry into isolation and to stay cautious after exiting isolation allows to achieve the same level of disease circulation with fewer people complying, thus helping to relieve the social burden imposed by the implementation of the self-isolation measure.
\par
In this work we also address the well-known trade-off between the benefits of longer isolation periods and of widespread compliance in breaking chains of transmission. The imitation dynamics being governed by this very trade-off, we are interested in seeing the conditions for its resolution within our framework.
We find that there exists a duration able to minimise the equilibrium prevalence, given by the longest isolation duration that still allows full compliance of the population, at a given cost of quarantine. The existence of this minimum implies that it is essential to reach high levels of adoption, even if this means that the isolation duration needs to be shortened. However, if the cost of isolation is too small, it is always beneficial to have the longest isolation duration, while if the cost is too high, increasing the isolation duration further than a certain point brings no benefit at all. Still, as expected, the equilibrium prevalence is always the smaller the lower the cost of quarantine. To summarise, in many settings, it is beneficial to shorten the isolation duration, while a lower cost of quarantine is always advantageous. Reducing the cost of quarantine can be achieved through incentives such as paid sick leave.
\par
Although the parameter $p$ that indicates the relative importance given to prevalence compared to the incidence rate plays no role on the equilibria nor on their stability conditions, it turns out that it may affect the transient phase of the dynamics. More specifically, either different choices of $p$ give the same dynamics, or a lower $p$ systematically gives a lower peak prevalence
in the transient phase. We measure this effect both in the mean-field framework and on a multiplex network. We find that favouring the incidence rate as an indicator of the state of an outbreak might help limit the amplitude of an initial overshoot. Additionally, the observed advantage of incidence over prevalence is the stronger, the faster the dynamics, i.e. for a larger reproduction number and a larger volatility of opinions. 
\par
We test our analytical predictions against numerical simulations on an empirical multiplex network from the Copenhagen Network Study. Our mean-field results hold fairly.

\vspace{1em}

This study comes with a set of limitations including: (i) the minimal design of a compartmental model counting a small number of compartments, (ii) the assumptions regarding how adoption of self-isolation spreads within the population, (iii) the representation of how individuals access information about the state of the outbreak, and (iv) the preliminary nature of our investigation of the impact of network structure on our predictions.

% \vspace{1em}

First off, for the sake of simplicity and so that we could draw analytically tractable conclusions, we have decided to work with a compartmental model -~that comes with its own limitations~- and we have opted for a limited amount of compartments.
Being a compartmental model, we have assumed a constant and homogeneous infectiousness and susceptibility of individuals. To partially account for a lower pathogen load at the end of the infectious period, we have however incorporated a post-quarantine reduced infectiousness. This reduced infectiousness though is missing in the course of infection of individuals who do not comply. The rates of transition from one compartment to the next are assumed to be Poissonian, despite evidence that, for example, the infectious period follows more complex distributions~\cite{He_2020,Marc_2021}. A more detailed modelling of the time spent in isolation could be obtained using age-structured partial differential equations (PDE), for instance considering the time elapsed since the start of quarantine~\cite{Wu_2023} or structuring the population with respect to the pathogen load~\cite{Della_Marca_2023}.
Apart from limitations affecting compartmental models in general, we have opted for a limited amount of compartments to keep analytical derivations tractable and interpretable.
For example, we have decided to disregard the isolation of healthy individuals. Such a situation would be interesting to account for a more realistic toll of the isolation measure. Healthy individuals may indeed be asked to isolate from the community when preventive quarantine is implemented (e.g. because of the lack of accurate, widely distributed and affordable diagnostic tests) instead of the self-isolation of only known cases treated in the present work. Moreover, we have not considered isolation periods that exceed the time for recovery, needlessly but as a consequence of the difficulty in estimating the end of the infectious period.
In addition, in order to better describe the evolution of infectious diseases such as COVID-19, an exposed compartment and a transitory removed compartment accounting for temporary immunity are generally considered to be important~\cite{Tori_2022, Zhang_2023}. In this light, our model corresponds to a worst case scenario.
We also disregarded the possibility for individuals to leave the dynamics as a consequence of death or of permanent measures such as vaccination.

% \vspace{1em}
Secondly, we have made some assumptions, that should be examined in future work, on how information on the state of the outbreak is accessed. To keep the model simple, the information about the state of the epidemic has been treated as memoryless, although some works have relaxed this assumption~\cite{Onofrio2007,Kyrychko_2023}.
We have assumed that individuals access real instantaneous values of the prevalence and of the incidence rate and we have completely disregarded delays or deviations from the actual values due e.g., to cognitive biases and underdetection. If the underdetection is systematic we would expect it to affect, in our model, the effective rate of change of opinion and the cost of quarantine.
Moreover, the addition of a compartment for hospitalised patients would allow to account also for the number of hospitalisations as an indicator of the severity of the outbreak and the saturation of health care~\cite{Verma_2025}.
Distinguishing between local and global information about the state of an outbreak would be interesting for a future analysis on networks~\cite{Cascante22, Silva_2023} or using a PDE approach~\cite{Banerjee_2024}. Indeed, these frameworks could account for biases arising from estimations of the severity or state of an outbreak based on observations made in people's local social environment and compare the effect of local and population-level information.

% \vspace{1em}

Thirdly, our analysis has limitations concerning the way compliance is assumed to evolve in time.
There is evidence that complex contagion better describes the propagation of opinions on social networks~\cite{Monsted17, Sprague17}. In our model, exposure to a single individual with opposite strategy is sufficient to allow someone to challenge its strategy.
Our model also disregards homophily, in particular homophily in the network of social influence~\cite{He2024}. We expect a threshold type of model and the introduction of homophily to reduce the parameter $\kappa$ that accounts for the strategy update rate, but the phenomenology could change in more complex ways.  We did not consider that individuals could spontaneously change their mind regarding the compliance with the control measure~\cite{Poletti_2009}. Also, as briefly discussed here above, we have considered that a change of opinion only affects individuals who are yet to decide whether to isolate or not. As a consequence, isolating individuals who change opinion are assumed not to exit quarantine immediately. Instead, it would make sense to consider that individuals who change their mind during the isolation period could abort it. We expect in this case that sustained oscillations would no longer be possible, similarly to what was found for mask wearers \cite{Martin25}. 

% \vspace{1em}
Finally, we have only performed a preliminary comparison of analytical predictions against stochastic simulations run on an empirical network. No systematic analysis of the impact of different network structures on the transient and stationary parts of the dynamics, for different values of the parameters, has been carried out, yet. Our results suggest that the network structure affects our analytical predictions in the stationary limit and may affect the amplitude of observed oscillations (potentially as a consequence of weak multi-layer degree correlations). It would be important in the future to check what network features are responsible for this discrepancy and whether specific network structures are able to affect the dependence of population-level quantities on the model parameters.

\section*{Acknowledgements}

We thank PEPS-CNRS for funding the project, François Castella for his critical look on the model, as well as Institut Agro and the Irmar lab (France) for inviting GdM. HM is grateful to Gerardo I\~niguez and the Computing Science lab of Tampere University (Finland) for their hospitality. 
HM acknowledges support of the `Chair Modélisation Mathématique et Biodiversité (MMB)' of Veolia - Ecole polytechnique - Museum national d’Histoire naturelle - Fondation X, and of Sanofi through the FluCov project. Finally, the authors are thankful to the organising team of the event Current
Challenges Workshop on Epidemic Modelling, Girona 2023, where they met and started to discuss the present work.

\section*{Declaration of interests}
All authors declare no conflicts of interest in this paper.

\newpage

\printbibliography

@article{Banerjee_2024,
title = {Behavior-induced phase transitions with far from equilibrium patterning in a SIS epidemic model: Global vs non-local feedback},
journal = {Physica D: Nonlinear Phenomena},
volume = {469},
pages = {134316},
year = {2024},
issn = {0167-2789},
doi = {https://doi.org/10.1016/j.physd.2024.134316},
url = {https://www.sciencedirect.com/science/article/pii/S0167278924002677},
author = {Malay Banerjee and Vitaly Volpert and Piero Manfredi and Alberto d’Onofrio}
}

@article {QSmet,
	author = {Silvio C. Ferreira and Claudio Castellano and Romualdo Pastor-Satorras},
	title = {Epidemic thresholds of the susceptible-infected-susceptible model on networks: A comparison of numerical and theoretical results},
	volume = {86},
	year = {2012},
	doi = {https://doi.org/10.1103/PhysRevE.86.041125},
	journal = {Physical Review E}
}

@article {Cota2017,
	author = {Wesley Cota and Silvio C. Ferreira},
	title = {Optimized Gillespie algorithms for the simulation of Markovian epidemic processes on large and heterogeneous networks},
	volume = {219},
	year = {2017},
	doi = {https://doi.org/10.1016/j.cpc.2017.06.007},
	URL = {https://www.sciencedirect.com/science/article/pii/S0010465517301893?via\%3Dihub},
	journal = {Computer Physics Communications}
}

@article {Netzschleuder_repo,
	author = {Tiago P. Peixoto},
	title = {The Netzschleuder network catalogue and repository},
	year = {2020},
	doi = {10.5281/zenodo.7839981},
	URL = {https://networks.skewed.de/},
}

@article {deMeijere2021,
	author = {Giulia de Meijere and Vittoria Colizza and Eugenio Valdano and Claudio Castellano},
	title = {Effect of delayed awareness and fatigue on the efficacy of self-isolation in epidemic control},
	volume = {104},
	year = {2021},
	doi = {10.1103/PhysRevE.104.044316},
	URL = {https://pubmed.ncbi.nlm.nih.gov/34781485/},
	journal = {Physical Review E}
}

@article{Wang_2015,
title = {Coupled disease–behavior dynamics on complex networks: A review},
journal = {Physics of Life Reviews},
volume = {15},
pages = {1-29},
year = {2015},
issn = {1571-0645},
doi = {https://doi.org/10.1016/j.plrev.2015.07.006},
url = {https://www.sciencedirect.com/science/article/pii/S1571064515001372},
author = {Zhen Wang and Michael A. Andrews and Zhi-Xi Wu and Lin Wang and Chris T. Bauch}
}

@article{Bauch_2005,
    author  = "Chris T. Bauch",
    title   = "Imitation dynamics predict vaccinating behaviour",
    year    = "2005",
    journal = "Proc. R. Soc. B.",
    volume  = "272",
    pages   = "1669–1675"
}

@article{Ndeffo_Mbah_2012,
    author  = {M. Ndeffo Mbah and J. Liu and C. Bauch and Y. Tekel and J. Medlock and L. Ancel Meyers and A. Galvani},
    title   = "The Impact of Imitation on Vaccination Behavior in Social Contact Networks",
    year    = "2012",
    journal = "PLoS Comput Biol",
    volume  = "8",
}

@article{Poletti_2009,
    author  = {P. Poletti and B. Caprile and M. Ajelli and A. Pugliese and S. Merler},
    title   = "Spontaneous behavioural changes in response to epidemics",
    year    = "2009",
    journal = "Journal of Theoretical Biology",
    volume  = "260",
}

@book{Allen_2007,
  title={An Introduction to Mathematical Biology},
  author={Allen, L.J.S.},
  isbn={9780130352163},
  lccn={2006042585},
  url={https://books.google.fr/books?id=e7RxQgAACAAJ},
  year={2007},
  publisher={Pearson/Prentice Hall}
}

@article{Petherick21,
title = {A worldwide assessment of changes in adherence to COVID-19 protective behaviours and hypothesized pandemic fatigue},
author = {Anna Petherick and Rafael Goldszmidt and Eduardo B. Andrade and Rodrigo Furst and Thomas Hale and Annalena Pott and Andrew Wood},
year = {2021},
journal = {Nature Human Behaviour},
volume = {5},
pages = {1145–1160},
doi = {https://doi.org/10.1038/s41562-021-01181-x}
}

@article{Martin25,
author = {Hugo Martin and François Castella and Frédéric Hamelin},
journal = {submitted for publication},
title = {Wearing face masks to protect oneself and/or others:
Counter-intuitive results from a simple epidemic model
accounting for selfish and altruistic human behaviour},
year = {2025},
doi = {https://hal.science/hal-05015410v1}}

@book{Anderson91,
author = {R. M. Anderson and R. M. May},
year = {1991},
title = {Infectious Diseases of Humans: Dynamics and Control}, 
journal = {Oxford University Press},
ISBN = {9780198540403}
}

@article{Hethcote2002,
title = {Effects of quarantine in six endemic models for infectious diseases},
journal = {Mathematical Biosciences},
volume = {180},
number = {1},
pages = {141-160},
year = {2002},
issn = {0025-5564},
doi = {10.1016/S0025-5564(02)00111-6},
author = {Herbert Hethcote and Ma Zhien and Liao Shengbing}
}

@article{Young2019,
title = {Consequences of delays and imperfect implementation of isolation in epidemic control},
journal = {Scientific Reports},
number = {2045-2322},
year = {2019},
doi = {10.1038/s41598-019-39714-0},
author = {Young, Lai-Sang and Ruschel, Stefan and Yanchuk, Serhiy and Pereira, Tiago}
}

@article{Chen2019,
  title={Global stability of epidemic models with imperfect vaccination and quarantine on scale-free networks},
  author={Chen, Shanshan and Small, Michael and Fu, Xinchu},
  journal={IEEE Transactions on Network Science and Engineering},
  volume={7},
  number={3},
  pages={1583--1596},
  year={2019},
  publisher={IEEE}
}

@article{Zhang2017,
title = {The threshold of a stochastic SIQS epidemic model},
journal = {Physica A: Statistical Mechanics and its Applications},
volume = {482},
pages = {362-374},
year = {2017},
issn = {0378-4371},
doi = {10.1016/j.physa.2017.04.100},
author = {Xiao-Bing Zhang and Hai-Feng Huo and Hong Xiang and Qihong Shi and Dungang Li}
}

@article{Esquivel2018,
  title={Efficiency of quarantine and self-protection processes in epidemic spreading control on scale-free networks},
  author={Esquivel-G{\'o}mez, Jose de Jesus and Barajas-Ram{\'\i}rez, Juan Gonzalo},
  journal={Chaos: An Interdisciplinary Journal of Nonlinear Science},
  volume={28},
  number={1},
  year={2018},
  publisher={AIP Publishing}
}

@article{Mancastroppa2020,
  title={Active and inactive quarantine in epidemic spreading on adaptive activity-driven networks},
  author={Mancastroppa, Marco and Burioni, Raffaella and Colizza, Vittoria and Vezzani, Alessandro},
  journal={Physical Review E},
  volume={102},
  number={2},
  pages={020301},
  year={2020},
  publisher={APS}
}

@article{Cascante22,
title = {How disease risk awareness modulates transmission: coupling infectious disease models with behavioural dynamics},
author = {Jaime Cascante-Vega and Samuel Torres-Florez and Juan Cordovez
and Mauricio Santos-Vega},
journal = {Royal Society Open Science},
volume = {9},
year = {2022},
doi = {https://doi.org/10.1098/rsos.210803}
}

@article{Onofrio2007,
title = {Vaccinating behaviour, information, and the dynamics of SIR vaccine preventable diseases},
author = {Alberto d’Onofrio and Piero Manfredi and Ernesto Salinelli},
journal = {Theoretical Population Biology},
year = {2007},
volume = {71},
issue = {3},
doi = {https://doi.org/10.1016/j.tpb.2007.01.001}}

@article{Montgomery_2020,
title = {Peer social network processes and adolescent health behaviors: A systematic review},
journal = {Preventive Medicine},
volume = {130},
pages = {105900},
year = {2020},
issn = {0091-7435},
doi = {https://doi.org/10.1016/j.ypmed.2019.105900},
url = {https://www.sciencedirect.com/science/article/pii/S0091743519303809},
author = {Shannon C. Montgomery and Michael Donnelly and Prachi Bhatnagar and Angela Carlin and Frank Kee and Ruth F. Hunter}
}

@article{Peak_2017,
author = {Corey M. Peak  and Lauren M. Childs  and Yonatan H. Grad  and Caroline O. Buckee },
title = {Comparing nonpharmaceutical interventions for containing emerging epidemics},
journal = {Proceedings of the National Academy of Sciences},
volume = {114},
number = {15},
pages = {4023-4028},
year = {2017},
doi = {10.1073/pnas.1616438114},
URL = {https://www.pnas.org/doi/abs/10.1073/pnas.1616438114},
eprint = {https://www.pnas.org/doi/pdf/10.1073/pnas.1616438114}
}

@Article{Verma_2025,
title = {Modelling behavioural interactions in infection disclosure during an outbreak: An evolutionary game theory approach},
journal = {Mathematical Biosciences and Engineering},
volume = {22},
number = {8},
pages = {1931-1955},
year = {2025},
issn = {1551-0018},
doi = {10.3934/mbe.2025070},
url = {https://www.aimspress.com/article/doi/10.3934/mbe.2025070},
author = {Pranav Verma and Viney Kumar and Samit Bhattacharyya},
}

@article{Kyrychko_2023,
    author = {Kyrychko, Y. N. and Blyuss, K. B.},
    title = {Vaccination games and imitation dynamics with memory},
    journal = {Chaos: An Interdisciplinary Journal of Nonlinear Science},
    volume = {33},
    number = {3},
    pages = {033134},
    year = {2023},
    month = {03},
    issn = {1054-1500},
    doi = {10.1063/5.0143184},
    url = {https://doi.org/10.1063/5.0143184},
    eprint = {https://pubs.aip.org/aip/cha/article-pdf/doi/10.1063/5.0143184/18188333/033134\_1\_5.0143184.pdf},
}

@article{Asfaw_2018,
	author = {Manalebish Asfaw and Bruno Buonomo and Semu Kassa},
	title = {Impact of human behavior on ITNs control strategies to prevent the spread of vector borne diseases},
	journal = {Atti della Accademia Peloritana dei Pericolanti - Classe di Scienze Fisiche, Matematiche e Naturali},
	volume = {96},
	number = {S3},
	year = {2018},
	issn = {1825-1242},	pages = {2},	doi = {10.1478/AAPP.96S3A2},
	url = {https://cab.unime.it/journals/index.php/AAPP/article/view/AAPP.96S3A2}
}

@incollection{White_2017,
title = {5 - Mathematical Models in Infectious Disease Epidemiology},
editor = {Jonathan Cohen and William G. Powderly and Steven M. Opal},
booktitle = {Infectious Diseases (Fourth Edition)},
publisher = {Elsevier},
edition = {Fourth Edition},
pages = {49-53.e1},
year = {2017},
isbn = {978-0-7020-6285-8},
doi = {https://doi.org/10.1016/B978-0-7020-6285-8.00005-8},
url = {https://www.sciencedirect.com/science/article/pii/B9780702062858000058},
author = {Peter J. White}
}

@article{Kabir_2021,
author = {Kabir, Ariful and Risa, Tori and Tanimoto, Jun},
title = {Prosocial behavior of wearing a mask during an epidemic: an evolutionary explanation},
journal = {Scientific Reports},
volume = {11},
pages = {12621},
year = {2021},
doi = {10.1038/s41598-021-92094-2},
URL = {https://doi.org/10.1038/s41598-021-92094-2}
}

@article{Martcheva_2021,
author = {Maia Martcheva, Necibe Tuncer and Calistus N. Ngonghala},
title = {Effects of social-distancing on infectious disease dynamics: an evolutionary game theory and economic perspective},
journal = {Journal of Biological Dynamics},
volume = {15},
number = {1},
pages = {342-366},
year = {2021},
publisher = {Taylor $\&$ Francis},
doi = {10.1080/17513758.2021.1946177},
note ={PMID: 34182892},
URL = {https://doi.org/10.1080/17513758.2021.1946177},
}

@article{Aghaeeyan_2024,
   title={Revealing Decision-Making Strategies of Americans in Taking COVID-19 Vaccination},
   volume={86},
   ISSN={1522-9602},
   url={http://dx.doi.org/10.1007/s11538-024-01290-4},
   DOI={10.1007/s11538-024-01290-4},
   number={6},
   journal={Bulletin of Mathematical Biology},
   publisher={Springer Science and Business Media LLC},
   author={Aghaeeyan, Azadeh and Ramazi, Pouria and Lewis, Mark A.},
   year={2024},
   month=may
}

@ARTICLE{Raude_2020,
  
AUTHOR={Raude, Jocelyn  and Lecrique, Jean-Michel  and Lasbeur, Linda  and Leon, Christophe  and Guignard, Romain  and du Roscoät, Enguerrand  and Arwidson, Pierre },
         
TITLE={Determinants of Preventive Behaviors in Response to the COVID-19 Pandemic in France: Comparing the Sociocultural, Psychosocial, and Social Cognitive Explanations},
        
JOURNAL={Frontiers in Psychology},
        
VOLUME={Volume 11 - 2020},

YEAR={2020},

URL={https://www.frontiersin.org/journals/psychology/articles/10.3389/fpsyg.2020.584500},

DOI={10.3389/fpsyg.2020.584500},

ISSN={1664-1078}
}

@article{Peretti-Watel_2020,
  title    = {Attitudes about {COVID-19} Lockdown among General Population, France, March 2020},
  author   = {Peretti-Watel, Patrick and Seror, Val{\'e}rie and Cortaredona, S{\'e}bastien and Launay, Odile and Raude, Jocelyn and Verger, Pierre and Beck, Fran{\c c}ois and Legleye, St{\'e}phane and L'Haridon, Olivier and Ward, Jeremy and {Confinement, Coronavirus And and Longitudinale, Enqu{\^e}te} and {COCONEL} and {Study Group}},
  journal  = {Emerg Infect Dis},
  volume   =  {27},
  number   =  {1},
  pages    = {301--303},
  year     =  {2020}
}

@article{Setbon_2010,
    author = {Setbon, Michel and Raude, Jocelyn},
    title = {Factors in vaccination intention against the pandemic influenza A/H1N1},
    journal = {European Journal of Public Health},
    volume = {20},
    number = {5},
    pages = {490-494},
    year = {2010},
    month = {05},
    issn = {1101-1262},
    doi = {10.1093/eurpub/ckq054},
    url = {https://doi.org/10.1093/eurpub/ckq054},
    eprint = {https://academic.oup.com/eurpub/article-pdf/20/5/490/1668273/ckq054.pdf},
}

@article{Sapiezynski19,
title= {Interaction data from the Copenhagen Networks Study},
author = {Piotr Sapiezynski and Arkadiusz Stopczynski and David Dreyer Lassen and Sune Lehmann},
journal = {Scientific Data},
year = {2019},
doi = {https://doi.org/10.1038/s41597-019-0325-x}
}

@BOOK{book_Tanimoto,
  title     = "Sociophysics approach to epidemics",
  author    = "Tanimoto, Jun",
  publisher = "Springer",
  series    = "Evolutionary Economics and Social Complexity Science",
  edition   =  2021,
  month     =  mar,
  year      =  2021,
  address   = "Singapore, Singapore",
  language  = "en"
}

@article{Bartolo19,
title = {Determinants of influenza vaccination uptake in pregnancy: a large single-Centre cohort study},
author = {Stéphanie Bartolo and Emilie Deliege and Ophélie Mancel and Philippe Dufour and Sophie Vanderstichele and Marielle Roumilhac and Yamina Hammou and Sophie Carpentier, Rodrigue Dessein and Damien Subtil and Karine Faure},
journal = {BMC Pregnancy and Childbirth},
volume = {19},
year = {2019},
doi = {https://doi.org/10.1186/s12884-019-2628-5}
}

@article{Betsch10,
author = {Cornelia Betsch and Frank Renkewitz and Tilmann Betsch and Corina Ulshöfer},
title = {The influence of vaccine-critical websites on perceiving vaccination risks},
journal = {Journal of Health Psychology},
year= {2010},
doi = {10.1177/1359105309353647}
}

@article{Bedson_2021,
author = {Jamie Bedson and Laura Skrip and Danielle Pedi and Sharon Abramowitz and Simone Carter and Mohamed Jalloh and Sebastian Funk and Nina Gobat and Tamara Giles-Vernick and Gerardo Chowell and Jo\~{a}o Rangel de Almeida and Rania Elessawi and Samuel Scarpino and Ross Hammond and Sylvie Briand and Joshua Epstein and Laurent Hébert-Dufresne and Benjamin M. Althouse},
title = {A review and agenda for integrated disease models including social and behavioural factors},
journal = {Nat Hum Behav},
year = {2021},
doi = {10.1038/s41562-021-01136-2}
}

@article{Funk_2015,
title = {Nine challenges in incorporating the dynamics of behaviour in infectious diseases models},
journal = {Epidemics},
volume = {10},
pages = {21-25},
year = {2015},
note = {Challenges in Modelling Infectious DIsease Dynamics},
issn = {1755-4365},
doi = {https://doi.org/10.1016/j.epidem.2014.09.005},
url = {https://www.sciencedirect.com/science/article/pii/S1755436514000541},
author = {Sebastian Funk and Shweta Bansal and Chris T. Bauch and Ken T.D. Eames and W. John Edmunds and Alison P. Galvani and Petra Klepac},
}

@Article{Funk_2010_review,
  author    = {S. Funk and M. Salathé and V. Jansen},
  title     = {Modelling the influence of human behaviour on the spread of infectious diseases: a review},
  journal   = {J. R. Soc. Interface },
  year      = {2010},
  volume    = {7},
  doi       = {10.1098/rsif.2010.0142},
  publisher = {The Royal Society},
}

@article{Verelst_2016,
author = {Verelst, Frederik  and Willem, Lander  and Beutels, Philippe },
title = {Behavioural change models for infectious disease transmission: a systematic review (2010–2015)},
journal = {Journal of The Royal Society Interface},
volume = {13},
number = {125},
pages = {20160820},
year = {2016},
doi = {10.1098/rsif.2016.0820},
URL = {https://royalsocietypublishing.org/doi/abs/10.1098/rsif.2016.0820},
eprint = {https://royalsocietypublishing.org/doi/pdf/10.1098/rsif.2016.0820}
}

@article{Chang_2020,
author = {Sheryl L. Chang and Mahendra Piraveenan and Philippa Pattison and Mikhail Prokopenko},
title = {Game theoretic modelling of infectious disease dynamics and intervention methods: a review},
journal = {Journal of Biological Dynamics},
volume = {14},
number = {1},
pages = {57-89},
year  = {2020},
publisher = {Taylor $\&$ Francis},
doi = {10.1080/17513758.2020.1720322},
note ={PMID: 31996099},
URL = {https://doi.org/10.1080/17513758.2020.1720322},
eprint = {https://doi.org/10.1080/17513758.2020.1720322}
}

@article{Hamilton_2024,
title = {Incorporating endogenous human behavior in models of COVID-19 transmission: A systematic scoping review},
journal = {Dialogues in Health},
volume = {4},
pages = {100179},
year = {2024},
issn = {2772-6533},
doi = {https://doi.org/10.1016/j.dialog.2024.100179},
url = {https://www.sciencedirect.com/science/article/pii/S2772653324000157},
author = {Alisa Hamilton and Fardad Haghpanah and Alexander Tulchinsky and Nodar Kipshidze and Suprena Poleon and Gary Lin and Hongru Du and Lauren Gardner and Eili Klein},
}

@article{Lejeune_2025,
title = {Formulating human risk response in epidemic models: Exogenous vs endogenous approaches},
journal = {European Journal of Operational Research},
year = {2025},
issn = {0377-2217},
doi = {https://doi.org/10.1016/j.ejor.2025.01.004},
url = {https://www.sciencedirect.com/science/article/pii/S0377221725000049},
}

@article{Reitenbach_2025,
doi = {10.1088/1361-6633/ad90ef},
url = {https://dx.doi.org/10.1088/1361-6633/ad90ef},
year = {2024},
publisher = {IOP Publishing},
volume = {88},
number = {1},
pages = {016601},
author = {Reitenbach, Andreas and Sartori, Fabio and Banisch, Sven and Golovin, Anastasia and Calero Valdez, André and Kretzschmar, Mirjam and Priesemann, Viola and Mäs, Michael},
title = {Coupled infectious disease and behavior dynamics. A review of model assumptions},
journal = {Reports on Progress in Physics},
}

@article{Oraby_2014,
author = {Oraby, Tamer  and Thampi, Vivek  and Bauch, Chris T. },
title = {The influence of social norms on the dynamics of vaccinating behaviour for paediatric infectious diseases},
journal = {Proceedings of the Royal Society B: Biological Sciences},
volume = {281},
number = {1780},
pages = {20133172},
year = {2014},
doi = {10.1098/rspb.2013.3172},

URL = {https://royalsocietypublishing.org/doi/abs/10.1098/rspb.2013.3172},
eprint = {https://royalsocietypublishing.org/doi/pdf/10.1098/rspb.2013.3172}
}

@Article{Chang_2019,
AUTHOR = {Chang, Sheryl Le and Piraveenan, Mahendra and Prokopenko, Mikhail},
TITLE = {The Effects of Imitation Dynamics on Vaccination Behaviours in SIR-Network Model},
JOURNAL = {International Journal of Environmental Research and Public Health},
VOLUME = {16},
YEAR = {2019},
NUMBER = {14},
ARTICLE-NUMBER = {2477},
URL = {https://www.mdpi.com/1660-4601/16/14/2477},
PubMedID = {31336761},
ISSN = {1660-4601},
DOI = {10.3390/ijerph16142477}
}

@article{Traulsen_2023,
author = {Arne Traulsen  and Simon A. Levin  and Chadi M. Saad-Roy },
title = {Individual costs and societal benefits of interventions during the COVID-19 pandemic},
journal = {Proceedings of the National Academy of Sciences},
volume = {120},
number = {24},
pages = {e2303546120},
year = {2023},
doi = {10.1073/pnas.2303546120},
URL = {https://www.pnas.org/doi/abs/10.1073/pnas.2303546120},
eprint = {https://www.pnas.org/doi/pdf/10.1073/pnas.2303546120}
}

@article{Silva_2023,
title = {Epidemic outbreaks with adaptive prevention on complex networks},
journal = {Communications in Nonlinear Science and Numerical Simulation},
volume = {116},
pages = {106877},
year = {2023},
issn = {1007-5704},
doi = {https://doi.org/10.1016/j.cnsns.2022.106877},
url = {https://www.sciencedirect.com/science/article/pii/S1007570422003641},
author = {Diogo H. Silva and Celia Anteneodo and Silvio C. Ferreira},
}

@article{Khan_2024,
title = {Influence of waning immunity on vaccination decision-making: A multi-strain epidemic model with an evolutionary approach analyzing cost and efficacy},
journal = {Infectious Disease Modelling},
volume = {9},
number = {3},
pages = {657-672},
year = {2024},
issn = {2468-0427},
doi = {https://doi.org/10.1016/j.idm.2024.03.004},
url = {https://www.sciencedirect.com/science/article/pii/S246804272400037X},
author = {Md. Mamun-Ur-Rashid Khan and Jun Tanimoto},
}

@article{Tori_2022,
title = {A study on prosocial behavior of wearing a mask and self-quarantining to prevent the spread of diseases underpinned by evolutionary game theory},
journal = {Chaos, Solitons $\&$ Fractals},
volume = {158},
pages = {112030},
year = {2022},
issn = {0960-0779},
doi = {https://doi.org/10.1016/j.chaos.2022.112030},
url = {https://www.sciencedirect.com/science/article/pii/S0960077922002405},
author = {Risa Tori and Jun Tanimoto}
}

@article{Yin_2022,
title = {Impact of co-evolution of negative vaccine-related information, vaccination behavior and epidemic spreading in multilayer networks},
journal = {Communications in Nonlinear Science and Numerical Simulation},
volume = {109},
pages = {106312},
year = {2022},
issn = {1007-5704},
doi = {https://doi.org/10.1016/j.cnsns.2022.106312},
url = {https://www.sciencedirect.com/science/article/pii/S1007570422000338},
author = {Qian Yin and Zhishuang Wang and Chengyi Xia and Chris T. Bauch}
}

@article{Khan_2023,
title = {Time delay of the appearance of a new strain can affect vaccination behavior and disease dynamics: An evolutionary explanation},
journal = {Infectious Disease Modelling},
volume = {8},
number = {3},
pages = {656-671},
year = {2023},
issn = {2468-0427},
doi = {https://doi.org/10.1016/j.idm.2023.06.001},
url = {https://www.sciencedirect.com/science/article/pii/S2468042723000477},
author = {Md. Mamun-Ur-Rashid Khan and Md. Rajib Arefin and Jun Tanimoto}
}

@article{Zhao_2020,
	author = {Shi Zhao and Lewi Stone and Daozhou Gao and Salihu S. Musa and Marc K. C. Chong and Daihai He and Maggie H. Wang},
	title = {Imitation dynamics in the mitigation of the novel coronavirus disease (COVID-19) outbreak in Wuhan, China from 2019 to 2020},
	journal = {Annals of Translational Medicine},
	volume = {8},
	number = {7},
	year = {2020},
	issn = {2305-5847},	url = {https://atm.amegroups.org/article/view/39932}
}

@article{Wu_2023,
title = {Modelling COVID-19 epidemic with confirmed cases-driven contact tracing quarantine},
journal = {Infectious Disease Modelling},
volume = {8},
number = {2},
pages = {415-426},
year = {2023},
issn = {2468-0427},
doi = {https://doi.org/10.1016/j.idm.2023.04.001},
url = {https://www.sciencedirect.com/science/article/pii/S246804272300026X},
author = {Fei Wu and Xiyin Liang and Jinzhi Lei}
}

@article{Zhang_2023,
author = {Zhang, Rongping and Xie, Boli and Kang, Yun and Liu, Maoxing},
title = {Modeling and analyzing quarantine strategies of epidemic on two-layer networks: game theory approach},
journal = {Journal of Biological Systems},
volume = {31},
number = {01},
pages = {21-35},
year = {2023},
doi = {10.1142/S021833902350002X},
URL = {
        https://doi.org/10.1142/S021833902350002X
},
eprint = { 
        https://doi.org/10.1142/S021833902350002X
}
}

@article {Marc_2021,
article_type = {journal},
title = {Quantifying the relationship between SARS-CoV-2 viral load and infectiousness},
author = {Marc, Aurélien and Kerioui, Marion and Blanquart, François and Bertrand, Julie and Mitjà, Oriol and Corbacho-Monné, Marc and Marks, Michael and Guedj, Jeremie},
editor = {Cobey, Sarah E and Van der Meer, Jos W},
volume = 10,
year = 2021,
pub_date = {2021-09-27},
pages = {e69302},
citation = {eLife 2021;10:e69302},
doi = {10.7554/eLife.69302},
url = {https://doi.org/10.7554/eLife.69302},
journal = {eLife},
issn = {2050-084X},
publisher = {eLife Sciences Publications, Ltd},
}

@ARTICLE{He_2020,
  title     = "Temporal dynamics in viral shedding and transmissibility of
               {COVID-19}",
  author    = "He, Xi and Lau, Eric H Y and Wu, Peng and Deng, Xilong and Wang,
               Jian and Hao, Xinxin and Lau, Yiu Chung and Wong, Jessica Y and
               Guan, Yujuan and Tan, Xinghua and Mo, Xiaoneng and Chen, Yanqing
               and Liao, Baolin and Chen, Weilie and Hu, Fengyu and Zhang, Qing
               and Zhong, Mingqiu and Wu, Yanrong and Zhao, Lingzhai and Zhang,
               Fuchun and Cowling, Benjamin J and Li, Fang and Leung, Gabriel M",
  journal   = "Nat. Med.",
  publisher = "Springer Science and Business Media LLC",
  volume    =  26,
  number    =  5,
  pages     = "672--675",
  month     =  may,
  year      =  2020,
  language  = "en"
}

@ARTICLE{Della_Marca_2023,
  title     = "An {SIR} model with viral load-dependent transmission",
  author    = "Della Marca, Rossella and Loy, Nadia and Tosin, Andrea",
  journal   = "J. Math. Biol.",
  publisher = "Springer Science and Business Media LLC",
  volume    =  86,
  number    =  4,
  pages     = "61",
  month     =  mar,
  year      =  2023,
  copyright = "https://creativecommons.org/licenses/by/4.0",
  language  = "en"
}

@article{Monsted17,
author = {Bjarke Mønsted and Piotr Sapieżyński and Emilio Ferrara and Sune Lehmann},
title = {Evidence of complex contagion of information in social media: An experiment using Twitter bots},
journal = {PLoS ONE},
doi = {https://doi.org/10.1371/journal.pone.0184148},
volume = {12},
year = {2017}}

@article{Sprague17,
author = {Daniel A. Sprague and Thomas House},
title = {Evidence for complex contagion models of social contagion from observational data},
journal = {PLoS ONE},
volume = {12},
year = {2017},
doi = {https://doi.org/10.1371/journal.pone.0180802}}

@Book{Hofbauer_Sigmund,
  title     = {The theory of evolution and dynamical systems : mathematical aspects of selection},
  publisher = {Cambridge University Press},
  year      = {1988},
  author    = {Josef Hofbauer, Karl Sigmund},
}

@article{Qiu2022,
  title = {Understanding the coevolution of mask wearing and epidemics: A network perspective},
  volume = {119},
  ISSN = {1091-6490},
  url = {http://dx.doi.org/10.1073/pnas.2123355119},
  DOI = {10.1073/pnas.2123355119},
  number = {26},
  journal = {Proceedings of the National Academy of Sciences},
  publisher = {Proceedings of the National Academy of Sciences},
  author = {Qiu,  Zirou and Espinoza,  Baltazar and Vasconcelos,  Vitor V. and Chen,  Chen and Constantino,  Sara M. and Crabtree,  Stefani A. and Yang,  Luojun and Vullikanti,  Anil and Chen,  Jiangzhuo and Weibull,  J\"{o}rgen and Basu,  Kaushik and Dixit,  Avinash and Levin,  Simon A. and Marathe,  Madhav V.},
  year = {2022},
  month = jun 
}

@article{Eksin2017,
  title = {Disease dynamics in a stochastic network game: a little empathy goes a long way in averting outbreaks},
  volume = {7},
  ISSN = {2045-2322},
  url = {http://dx.doi.org/10.1038/srep44122},
  DOI = {10.1038/srep44122},
  number = {1},
  journal = {Scientific Reports},
  publisher = {Springer Science and Business Media LLC},
  author = {Eksin,  Ceyhun and Shamma,  Jeff S. and Weitz,  Joshua S.},
  year = {2017},
  month = mar 
}

@article{He2024,
title = {Effect of homophily on coupled behavior-disease dynamics near a tipping point},
journal = {Mathematical Biosciences},
volume = {376},
pages = {109264},
year = {2024},
issn = {0025-5564},
doi = {https://doi.org/10.1016/j.mbs.2024.109264},
url = {https://www.sciencedirect.com/science/article/pii/S002555642400124X},
author = {Zitao He and Chris T. Bauch}
}

@article{Wang16,
title = {Statistical physics of vaccination},
journal = {Physics Reports},
volume = {664},
pages = {1-113},
year = {2016},
note = {Statistical physics of vaccination},
issn = {0370-1573},
doi = {https://doi.org/10.1016/j.physrep.2016.10.006},
url = {https://www.sciencedirect.com/science/article/pii/S0370157316303349},
author = {Zhen Wang and Chris T. Bauch and Samit Bhattacharyya and Alberto d'Onofrio and Piero Manfredi and Matjaž Perc and Nicola Perra and Marcel Salathé and Dawei Zhao}
}

@article{Weston18,
title = {Infection prevention behaviour and infectious disease modelling: a review of the literature and recommendations for the future},
journal = {BMC Public Health},
volume = {18},
pages = {336},
year = {2018},
doi = {https://doi.org/10.1186/s12889-018-5223-1},
author = {Dale Weston and Katharina Hauck and Richard Amlôt}
}

@book{Onofriobook,
author = {D'Onofrio, Alberto and Manfredi, Piero},
year = {2013},
month = {10},
pages = {},
title = {Modeling the Interplay Between Human Behavior and the Spread of Infectious Diseases},
isbn = {978-1-4614-5474-8},
doi = {10.1007/978-1-4614-5474-8}
}

@article{Bliman22,
author = {Bliman, Pierre-Alexandre and Carrozzo-Magli, Alessio and D'Onofrio, Alberto and Manfredi, Piero},
year = {2022},
month = {12},
pages = {},
title = {Tiered social distancing policies and epidemic control},
volume = {478},
journal = {Proceedings of the Royal Society A: Mathematical, Physical and Engineering Sciences},
doi = {10.1098/rspa.2022.0175}
}

@article{Alimohamadi2020,
  title = {Estimate of the Basic Reproduction Number for COVID-19: A Systematic Review and Meta-analysis},
  volume = {53},
  ISSN = {2233-4521},
  url = {http://dx.doi.org/10.3961/jpmph.20.076},
  DOI = {10.3961/jpmph.20.076},
  number = {3},
  journal = {Journal of Preventive Medicine and Public Health},
  publisher = {Korean Society for Preventive Medicine},
  author = {Alimohamadi,  Yousef and Taghdir,  Maryam and Sepandi,  Mojtaba},
  year = {2020},
  month = may,
  pages = {151–157}
}

@article{vandenDriessche2017,
  title = {Reproduction numbers of infectious disease models},
  volume = {2},
  ISSN = {2468-0427},
  url = {http://dx.doi.org/10.1016/j.idm.2017.06.002},
  DOI = {10.1016/j.idm.2017.06.002},
  number = {3},
  journal = {Infectious Disease Modelling},
  publisher = {Elsevier BV},
  author = {van den Driessche,  Pauline},
  year = {2017},
  month = aug,
  pages = {288–303}
}

@article{Leith_2021,
author = {Leith, David and L’Orange, Christian and Volckens, John},
title = {Quantitative Protection Factors for Common Masks and Face Coverings},
journal = {Environmental Science \& Technology},
volume = {55},
number = {5},
pages = {3136-3143},
year = {2021},
doi = {10.1021/acs.est.0c07291},
note ={PMID: 33601881},
URL = {https://doi.org/10.1021/acs.est.0c07291},
eprint = {https://doi.org/10.1021/acs.est.0c07291}
}

@article{Koh_2022,
title = {Outward and inward protection efficiencies of different mask designs for different respiratory activities},
journal = {Journal of Aerosol Science},
volume = {160},
pages = {105905},
year = {2022},
issn = {0021-8502},
doi = {https://doi.org/10.1016/j.jaerosci.2021.105905},
url = {https://www.sciencedirect.com/science/article/pii/S0021850221006303},
author = {Xue Qi Koh and Anqi Sng and Jing Yee Chee and Anton Sadovoy and Ping Luo and Dan Daniel},
}

@article{Amaral_2021,
title = {An epidemiological model with voluntary quarantine strategies governed by evolutionary game dynamics},
journal = {Chaos, Solitons $\&$ Fractals},
volume = {143},
pages = {110616},
year = {2021},
issn = {0960-0779},
doi = {https://doi.org/10.1016/j.chaos.2020.110616},
url = {https://www.sciencedirect.com/science/article/pii/S0960077920310079},
author = {Marco A. Amaral and Marcelo M. de Oliveira and Marco A. Javarone}
}

\newpage
\appendix

\section{Recovery from U, I, Q and F occurs with the same rate\label{sec:rate_mu}}

Individuals in compartment $Q$ can recover through two \textit{alternative} paths: $Q \rightarrow S$ and $Q \rightarrow F \rightarrow S$. Let $t_S$ be the time for the event $Q \rightarrow S$ and $t_F$ the time for the event $Q \rightarrow F$. The probability that the path $Q \rightarrow S$ is chosen is:
\[
    P(t_S < t_F) = P(t_S = t) P(t_F > t),
\]
where $P(t_S = t) = \mu e^{-\mu t}$ and 
\[
    P(t_F > t) = \int_{t}^{\infty} \mu_Q e^{-\mu_Q s} \d s = e^{-\mu_Q t}.
\]
Similarly, the probability that the path $Q \rightarrow F$ is chosen is:
\[
    P(t_F < t_S) = \mu_Q e^{-\mu_Q t} e^{-\mu t}.
\]
The normalisation of $P(t_S < t_F)$ is 
\[
    N_S = \int_0^{\infty} \mu e^{-(\mu + \mu_Q) t}\d t = \frac{\mu}{\mu+\mu_Q}
\]
while, the normalisation of $P(t_F < t_S)$ is 
\[
    N_F = \int_0^{\infty} \mu_Q e^{-(\mu + \mu_Q) t} \d t = \frac{\mu_Q}{\mu+\mu_Q}.
\]
Thus, the average time taken to recover through the direct recovery path $Q \rightarrow S$ is:
\begin{align*}
    \tau_{Q\rightarrow S} &= \int_0^{\infty} t \frac{P(t_S = t) P(t_F > t)}{N_S} dt\\
    &=  \int_0^{\infty} (\mu+\mu_Q)t e^{-(\mu + \mu_Q)t} dt\\
    &= \frac{1}{\mu+\mu_Q},
\end{align*}
where in the last step we performed an integration by parts.
This average time is the same as the one of the transition $Q \rightarrow F$. The average time spent in isolation is thus also $\frac{1}{\mu+\mu_Q}$. This is because the sharpest exponential dominates the other one, whether it be the one of the $Q \rightarrow S$ or of the $Q \rightarrow F$ transition. The probability of the $Q \rightarrow S$ transition is:
\[
    p_{Q\rightarrow S} = \int_0^{\infty} P(t_S = t) P(t_F > t) dt= N_S= \frac{\mu}{\mu+\mu_Q}.
\]
For the $Q \rightarrow F$ transition, it is instead:
\[
    p_{Q \rightarrow F} = \int_0^{\infty} P(t_F = t) P(t_S > t) dt= N_F= \frac{\mu_Q}{\mu+\mu_Q}.
\]
So overall,
\begin{align*}
    \tau &= \tau_{Q\rightarrow S} p_{Q\rightarrow S} + \tau_{Q\rightarrow F\rightarrow S} p_{Q\rightarrow F\rightarrow S}= \frac{1}{\mu+\mu_Q} \frac{\mu}{\mu+\mu_Q} + \left(\frac{1}{\mu+\mu_Q} + \frac{1}{\mu}\right) \frac{\mu_Q}{\mu+\mu_Q}\\
    &=\frac{1}{\mu},
\end{align*}
the average time for the recovery of individuals in compartment $Q$ is simply the average time of a Poissonian transition with rate $\mu$. The same reasoning can be repeated for compartment $U$, that also shows multiple alternative paths to recovery. In this way, the average time of recovery of $\I = U + I + Q + F$ is also $1/\mu$, as expected for an SIS model.

\section{Stability of the fixed points\label{apdx:stab}}

We begin this section by recalling the Routh-Hurwitz criteria, specified for polynomials of order $3$ and $5$, that we will use to prove the local asymptotic stability of the EFC and EPC equilibria.

\subsection{Routh-Hurwitz criteria}

\begin{lemma}[Routh-Hurwitz Conditions for Orders 3 and 5~\cite{Allen_2007}]
Let the monic real polynomial \( P(\lambda) = \lambda^n + \alpha_1 \lambda^{n-1} + \cdots + \alpha_{n-1} \lambda + \alpha_n \) be given. Then \( P(\lambda) \) is Hurwitz stable—that is, all roots have strictly negative real parts—if and only if the following conditions hold:

\medskip
\textbf{For \( n = 3 \):}

\[
\begin{cases}
\alpha_1 > 0,\\
\alpha_3 > 0,\\
\alpha_1 \alpha_2 > \alpha_3.
\end{cases}
\]

\medskip
\textbf{For \( n = 5 \):}

\[
\begin{cases}
\alpha_i > 0, \text{ for i} \in \{1, 2, 3, 4, 5\}\\
\alpha_1 \alpha_2 \alpha_3 > \alpha_3^2 + \alpha_1^2 \alpha_4,\\
(\alpha_1 \alpha_4 - \alpha_5)(\alpha_1 \alpha_2 \alpha_3 - \alpha_3^2 - \alpha_1^2 \alpha_4) > \alpha_5(\alpha_1 \alpha_2 - \alpha_3)^2 + \alpha_1 \alpha_5^2.
\end{cases}
\]

These algebraic conditions are necessary and sufficient for the roots of \( P(\lambda) \) to lie entirely in the open left-half of the complex plane.
\end{lemma}

\subsection{Jacobian matrix}
To assess the stability of each endemic equilibrium, we compute the jacobian matrix $J_{\mathcal{E}} =$
\[
\left(
\begin{smallmatrix}
\frac{\RO}{\mathcal{R}_{x^*}} - \R_{x^*}  & 0 & -\frac{\RO}{\mathcal{R}_{x^*}} & -\epsilon \frac{\RO}{\mathcal{R}_{x^*}} & 0\\
\frac{\RO}{\mathcal{R}_{x^*}} - \R_{x^*} + 1 & -(1+u)  & -\frac{\RO}{\mathcal{R}_{x^*}} & -\epsilon \frac{\RO}{\mathcal{R}_{x^*}} & 0 \\
0 & u x^* & -(1+q) & 0 & \frac{u}{1+u}\left(1 - \frac{1}{\R_{x^*}}\right) \\
0 & 0 & q & -1 & 0 \\
\kappa x^*(1-x^*)\left[1 + (1 - p)\left(\frac{\RO}{\mathcal{R}_{x^*}} - \mathcal{R}_{x^*}\right)\right] & 0 & -\kappa x^*(1-x^*) (1 - p)\frac{\RO}{\mathcal{R}_{x^*}} & -\epsilon \kappa x^*(1-x^*) (1 - p)\frac{\RO}{\mathcal{R}_{x^*}} & \kappa (1 - 2x^*) \left[1 - \frac{1}{\mathcal{R}_{x^*}} - \frac{c}{1 + q} \right]
\end{smallmatrix}
\right)
\]
and specify its values based on $\mathcal{E} \in \left\{\mathrm{ENC, EPC, EFC}\right\}$, i.e. $x^* = 0$, $x^* = 1$, or in-between.

\subsection{Disease-free fixed points of the dynamics}
For the two disease-free equilibria, we recover classical conditions: they always exist, and DFNC is locally asymptotically stable if $\RO \leq 1$ (the special case $\RO = 1$ being dealt with by hand), while DFFC is never stable. 
\begin{comment}
Since the fourth column has only one non-zero element, the spectrum of the matrix $J_{\mathcal{E}}$ is the union of $\{-1\}$ and of the spectrum of $\tilde{J}_{\mathcal{E}}$, defined as $\tilde{J}_{\mathcal{E}} = $
\[
\left(
\begin{smallmatrix}
\RO(1 + Q^* - 2 \I^*) - 1  & 0 & -\RO (1 - \I^*) & 0 \\
\RO(1 + Q^* - 2 \I^*) x^* & -u  & -\RO (1 - \I^*) & 0 \\
0 & (u - 1)x^* & -q & \frac{\mu_U}{\mu}U^* \\
\kappa x^*(1-x^*)\left[p + (1 - p)(1 + Q^* - 2 \I^*)\right] & 0 & -\kappa x^* (1 - x^*) (1 - p)(1 - \I^*) & \kappa (1 - 2x^*) \left[p\I^* + (1 - p)(1 - \I^*)(\I^* - Q^*) - \frac{c}{\mu + \mu_Q} \right]
\end{smallmatrix}
\right).
\]
\end{comment}

\subsection{Endemic fixed points of the dynamics}

\subsubsection{ENC equilibrium}
For the ENC equilibrium, the jacobian matrix reduces to
\[
\begin{pmatrix}
1 - \RO  & 0 & -1 & - \epsilon & 0 \\
2 - \RO & -(1+u)  & -1 & - \epsilon & 0 \\
0 & 0 & -(1+q) & 0 & \frac{u}{1+u}\left(1 - \frac{1}{\RO}\right) \\
0 & 0 & q & -1 & 0 \\
0 & 0 & 0 & 0 & \kappa \left[1 - \frac{1}{\RO} - \frac{c}{1 + q} \right],
\end{pmatrix}
\]
which is a block-triangular matrix, so its associated characteristic polynomial is
\[
    P(\lambda) = (1 - \RO + \lambda)\left(1 + u + \lambda\right)(1 + \lambda)\left(1 + q + \lambda\right)\left(\kappa \left[1 - \frac{1}{\RO} - \frac{c}{1 + q} \right] - \lambda\right).
\]
It yields two conditions for local asymptotic stability, namely $\RO > 1$ and
\[
1 - \frac{1}{\RO} < \frac{c}{1 + q}.
\]
\subsubsection{EFC equilibrium}
For the EFC equilibrium instead the Jacobian matrix reads:
\[
\begin{pmatrix}
\frac{\RO}{\R_1} - \R_1  & 0 & -\frac{\RO}{\R_1} & -\epsilon\frac{\RO}{\R_1} & 0 \\
\frac{\RO}{\R_1} - \R_1 + 1 & -(1+u)  & -\frac{\RO}{\R_1} & -\epsilon\frac{\RO}{\R_1} & 0 \\
0 & u & -(1+q) & 0 & \frac{u}{1+u}\left(1-\frac{1}{\R_1}\right) \\
0 & 0 & q & -1 & 0 \\
0 & 0 & 0 & 0 & -\kappa \left[1-\frac{1}{\R_1} - \frac{c}{1 + q} \right].
\end{pmatrix}
\]
\begin{comment}
\[
\begin{pmatrix}
\R_1 - \frac{\RO}{\R_1}  & 0 & -\frac{\RO}{\R_1} & 0 \\
\frac{\RO}{\R_1} - \R_1 + 1 & -\left(1 + \frac{\mu_U}{\mu}\right)  & -\frac{\RO}{\R_1} & 0 \\
0 & \frac{\mu_U}{\mu} & -\left(1 + \frac{\mu_Q}{\mu}\right) & \frac{\mu_U}{\mu + \mu_U}\left(1-\frac{1}{\R_1}\right) \\
0 & 0 & 0 & -\kappa \left[\left(1-\frac{1}{\R_1}\right) \left(p + (1-p)\frac{1}{\RO}\right) - \frac{c}{\mu + \mu_Q} \right].
\end{pmatrix}
\]
\end{comment}
and developing with respect to the last row provides the stability condition
\[
\left(1-\frac{1}{\R_1}\right) > \frac{c}{1+q}.
\]
The rest of the spectrum is given by computing the determinant
\[
\begin{vmatrix}
\frac{\RO}{\R_1} -\R_1 -\lb & 0 & -\frac{\RO}{\R_1} & -\epsilon\frac{\RO}{\R_1} \\
\frac{\RO}{\R_1} - \R_1 + 1 & -(1+u+\lb)  & -\frac{\RO}{\R_1} & -\epsilon\frac{\RO}{\R_1} \\
0 & u & -(1+q+\lb) & 0 \\
0 & 0 & q & -1-\lb
\end{vmatrix}
\]
which can be processed for instance by applying the following operations
\begin{enumerate}
    \item L1 $\longleftarrow$ L1 - L2 - L3 - L4
    \item Factorization of $-(1 + \lb)$ in the first row
    \item $\forall j\{2, 3, 4\}$, C$j$ $\longleftarrow$ C$j$ + C1
    \item Developing along the first row
    \item L1 $\longleftarrow$ L1 - $\frac{1 - \R_1}{q}$L3
\end{enumerate}
to obtain
\begin{align*}
    &-(1+\lb)\times
\begin{vmatrix}
\R_1 - \frac{\RO}{\R_1} - u -\lb & 0 & (1 - \epsilon)\frac{\RO}{\R_1} + (1 - \R_1) + \frac{1 - \R_1}{q}(1 + \lb) \\
u & -(1+q+\lb) & 0 \\
0 & q & - (1 + \lb)
\end{vmatrix} \\
&= (1+\lb) \times \\
&\left[\left(\R_1 - \frac{\RO}{\R_1} + u + \lb\right)(1 + q + \lb)(1 + \lb) - uq\left((1 - \R_1) + (1 - \epsilon)\frac{\RO}{\R_1} + \frac{1 - \R_1}{q}(1 + \lb)\right) \right]. 
\end{align*}
using the rule of Sarrus.

We aim at applying the Routh-Hurwitz criterion to prove that all roots have a negative real part. Developing the polynomial in the brackets, we obtain $P(\lb) = \lb^3 + \alpha_1\lb^2 + \alpha_2\lb + \alpha_3$ with
\begin{itemize}
    \item[--] $\alpha_1 = \R_1 - \frac{\RO}{\R_1} + u + q$,
    \item[--] $\alpha_2 = \R_1 (u + q) - \frac{\RO}{\R_1}(q + 1) + uq = \R_1(u + q) - \frac{\RO}{\R_1}(q+1) + uq$,
    \item[--] $\alpha_3 = \left(\R_1 - \frac{\RO}{\R_1} + u - 1\right)q + (u-1)(q-1)\left[(\R_1 - 1)\frac{q}{q-1} - (1 - \epsilon)\frac{\RO}{\R_1} \right]$.
\end{itemize}
We need to check that $\alpha_1 > 0$, $\alpha_3 > 0$ and $\alpha_1 \alpha_2 > \alpha_3$. The first positivity is readily obtained by observing that
\[
\frac{\RO}{\R_1} = \frac{(u+1) (q+1)}{(u+1) (q+1) - u(1+\epsilon q)} \leq u+1
\]
combined with $\R_1 > 1$. The second is obtained by computing
\begin{align*}
    \alpha_3 &= \R_1 (u+1)(q+1) - \frac{\RO}{\R_1}(q+1 + (1 - \epsilon)u q)\\
    & = (u+1)(q+1)\left(\R_1 - \frac{q+1 + (1 - \epsilon)uq}{(u+1)(q+1) - u(1 + \epsilon q)}\right)\\
    & = (u+1)(q+1) (\R_1 - 1)
\end{align*}
so positivity is obtained by the condition $\R_1 > 1$. Finally, we show the last inequality be rewriting
\[
\alpha_1 = \left(u+1 - \frac{\RO}{\R_1}\right) + \R_1 + q +1 \qquad \text{and} \qquad \alpha_2 = (q+1)\R_1 + (q+2)\left(u + 1 - \frac{\RO}{\R_1}\right) + (u+1)(\R_1 - 1)
\]
which are both sums of three non-negative terms, with the product of the last of each being exactly $\alpha_3$.
\subsubsection{EPC equilibrium}
To assess local stability, we have to study the negativity of the roots of:{\footnotesize
\[
\begin{vmatrix}
\mathcal{R}_{x^*} - \frac{\RO}{\mathcal{R}_{x^*}} - \lb  & 0 & -\frac{\RO}{\mathcal{R}_{x^*}} & -\epsilon \frac{\RO}{\mathcal{R}_{x^*}} & 0 \\
1 + \mathcal{R}_{x^*} - \frac{\RO}{\mathcal{R}_{x^*}} & -(1 + u + \lb)  & -\frac{\RO}{\mathcal{R}_{x^*}} & -\epsilon \frac{\RO}{\mathcal{R}_{x^*}} & 0 \\
0 & u x^* & -(1 + q + \lb) & 0 & \frac{u}{1+u}\left(1 - \frac{1}{\mathcal{R}_{x^*}}\right) \\
0 & 0 & q & - (1 + \lb) & 0 \\
\kappa X\left[p + (1 - p)(1 - \mathcal{R}_{x^*} + \frac{\RO}{\mathcal{R}_{x^*}})\right] & 0 & -\kappa X \frac{1 - p}{\mathcal{R}_{x^*}} & -\epsilon \kappa X \frac{1 - p}{\mathcal{R}_{x^*}} & -\lb
\end{vmatrix}
\]
}
with $X = x^*(1-x^*) > 0$. We first perform
\begin{enumerate}
    \item C4 $\longleftarrow$ C4 - $\epsilon$ C3
    \item L2 $\longleftarrow$ L2 - L1
\end{enumerate}
and then develop with respect to the last column to obtain the associated characteristic polynomial (up to the sign)
\[
P(\lambda) = \lambda D_4(\lambda) + \kappa \frac{u}{1+u}\left(1 - \frac{1}{\mathcal{R}_{x^*}}\right) \frac{\RO}{\mathcal{R}_{x^*}} X D_3(\lambda)
\]
with $D_4(\lambda)$ and $D_3(\lambda)$ being the resulting $4 \times 4$ determinants, leading to polynomials in $\lambda$ of degree $4$ and $3$, respectively. Then we compute
\begin{align*}
    D_4(\lambda) &= \begin{vmatrix}
\mathcal{R}_{x^*} - \frac{\RO}{\mathcal{R}_{x^*}} - \lb  & 0 & -\frac{\RO}{\mathcal{R}_{x^*}} & 0 \\
1 + \lb & -(1 + u + \lb)  & 0 & 0 \\
0 & u x^* & -(1 + q + \lb) & \epsilon (1 + q + \lb) \\
0 & 0 & q & - (1 + \epsilon q + \lb)
\end{vmatrix}\\
    &= \begin{vmatrix}
\mathcal{R}_{x^*} - \frac{\RO}{\mathcal{R}_{x^*}} - \lb  & 0 & -\frac{\RO}{\mathcal{R}_{x^*}} & 0 \\
1 + \lb & -(1 + u + \lb)  & 0 & 0 \\
0 & u x^* & -(1 + \lb) & -(1 - \epsilon) (1 + \lb) \\
0 & 0 & q & - (1 + \epsilon q + \lb)
\end{vmatrix}\\
    &=(1+\lb)\bigl[\lb^3 + (\frac{\RO}{\mathcal{R}_{x^*}} - \mathcal{R}_{x^*}+u+q+2)\lb^2\\
    & \qquad \qquad + \left(\frac{\RO}{\mathcal{R}_{x^*}} u x^* + \left(\frac{\RO}{\mathcal{R}_{x^*}} - \mathcal{R}_{x^*}\right)(u+q+2) + (1+u)(1+q)\right)\lb\\
    & \qquad \qquad \qquad \qquad+ \frac{\RO}{\mathcal{R}_{x^*}} u x^*(1+\epsilon q) + \left(\frac{\RO}{\mathcal{R}_{x^*}} - \mathcal{R}_{x^*}\right)(1+u)(1+q)\bigr]
\end{align*}
and straightforwardly apply the Routh-Hurwitz criterion to the polynomial in the brackets. As a result, all roots of $D_4$ have negative real parts. Some additional computations lead to
\[
    D_3(\lb) = (1 - p)\frac{\RO}{\mathcal{R}_{x^*}}(1 + \epsilon q + \lb)\left(\frac{\RO}{\mathcal{R}_{x^*}} - \mathcal{R}_{x^*} +\lb\right)(1 + u + \lb)
\]
which has three negative roots. For any $\kappa >0$, all coefficients appearing in $P$ are positive, so for sufficiently small positive perturbation of $\lambda D_4(\lambda)$ by the other term of $P$, the two last conditions to check in the Routh-Hurwith criterion still holds. As a result, all roots of $P$ have negative real parts at least for sufficiently small $\kappa >0$. The code provided as a companion for this proof explicitly computes the critical bound above which the local asymptotic stability is not guaranteed. The code also displays two plots of the temporal dynamics of the $x$ variable: one for a $\kappa$ lying $5\%$ below the critical value and the other $5\%$ above. The former converges to a single value, whereas the latter exhibits sustained oscillations.

\section{Impact of imitation on the dynamics}
\label{section_imitation}

\subsection{Equilibrium prevalence}

When imitation is completely turned off our model collapses to the IDF model and the disease prevalence equation in the stationary limit reads

\[
    \I^*(p_Q) := 1 - \frac{1}{\mathcal R_{p_Q}},
\]
with
\[
    \mathcal R_{p_Q}:= \left(1-\frac{u}{1+u}\frac{1 + \epsilon q}{1 + q}p_Q\right)\RO
\]

which is analogous to Equation~\ref{eq:Rx}, if we have a fixed compliance $p_Q$ instead of the equilibrium compliance $x^*$ that depends on the model parameters.

This suggests that, whether a coupled disease and opinion dynamics would yield a higher or lower equilibrium prevalence compared to the same disease model without opinion dynamics depends uniquely on the comparison between $p_Q$ (the population baseline compliance) and $x^*$. In general, an imitation dynamics or self-regulation provides better epidemiological results as long as we are in the regions of parameter space where $x^* > p_Q$ is satisfied. When a population has very low adoption levels, it is better for people to get exposed to information such as prevalence and incidence, in order to have a chance to increase the overall compliance and improve mitigation. Instead, when a population has a very large baseline level of adoption, it is advantageous for people not to rethink their behaviour, in order to maintain high levels of compliance.

\subsection{Transient prevalence}

Although there is a discontinuity in the equilibrium prevalence between the limit case $\kappa = 0$ treated in the IDF model and the case $\kappa > 0$ studied in the present work, we see that this discontinuity arises in the following way from a common initial condition.

When $\kappa = 0$, the dynamics initially follows the trajectory predicted by the IDF model with $p_Q=x(t=0)$ given by the initial condition. However, as soon as $\kappa > 0$, the long-time regime is characterised by a convergence to the trajectory predicted by IDF with $p_Q = x^*$. Increasing $\kappa$ further, anticipates the moment in which the initial condition is forgotten by the system yielding its collapse onto its long-time trajectory. We thus find that despite the fact that the volatility $\kappa$ of opinions does not contribute to the equilibria, it affects the transient phase.

Moreover, if the initial condition $x_0 < x^*$ (or the reverse $x_0 > x^*$), the IDF model predicts $\I(p_Q = x_0) > \I(p_Q = x^*)$ (or the reverse $\I(p_Q = x_0) < \I(p_Q = x^*)$). This has a straighforward implication on the effect of $\kappa$ on the transient phase: if $x_0 < x^*$ (or the reverse $x_0 > x^*$) a lower cumulative number of infected individuals is achieved when the volatility $\kappa$ is high (low). 

See Figure~\ref{fig:effect_kappa} for a representative example of the effect of $\kappa$ on the transient phase of the dynamics.

\subsection{Optimal isolation duration with and without imitation dynamics}

In the model with imitation dynamics, the compliance takes a negative contribution given by the average duration of the isolation period (the payoff  $\pi_Q$ of isolating). In Figure~\ref{fig:no_imitation_dependence}, we compare the $\kappa >0$ case (main text) with the $\kappa = 0$ limit but such that the level of adoption is a decaying function of the isolation duration (similarly to~\cite{deMeijere2021}). This coupling between the level of adoption and the isolation duration allows the $q \rightarrow 0$ limit of very long quarantines to be associated with a larger prevalence than in the case without coupling. This yields an optimal isolation duration in all the parameter space when $\kappa=0$, which is not observed when $\kappa > 0$.

\begin{figure}
    \centering
    \includegraphics[width = 0.7 \linewidth]{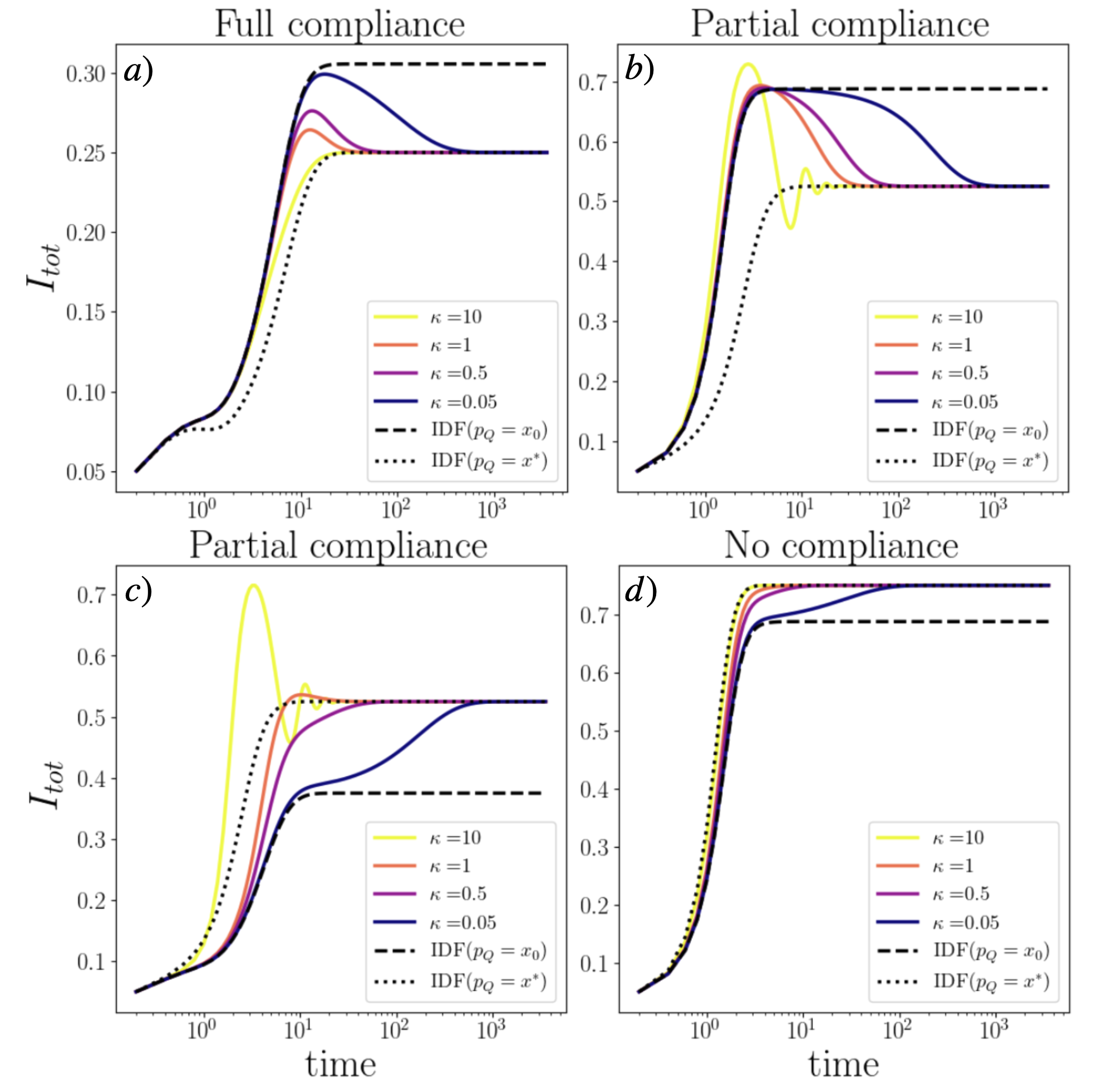}
    \caption{\textbf{Effect of $\kappa$ on the transient phase, in the EFC, EPC and ENC regimes}. Evolution of the prevalence in time with $x_0$, the initial condition on the fraction of cooperators. We distinguish between the scenarios with $x_0 < x^*$ (`Full compliance' and `Partial compliance') and those with $x_0 > x^*$ (`Partial compliance' and `No compliance'). Parameter values: $\mathcal{R}_0 = 4$, $\epsilon = 0.2$, $\mu=1$, $\mu_U = 5$, $\mu_Q = 1/3$. In `No Compliance': $x_0 = 0.3$ and $c=2$ ; in `Partial Compliance': $x_0 = 0.3$ and $c = 0.7$ (center, left) or $x_0 = 0.9$ and $c = 0.7$ (center, right) ; in `Full Compliance': $x_0 = 0.96$ and $c = 0.1$.}
    \label{fig:effect_kappa}
\end{figure}

\begin{figure}
    \centering
    \includegraphics[width=\textwidth]{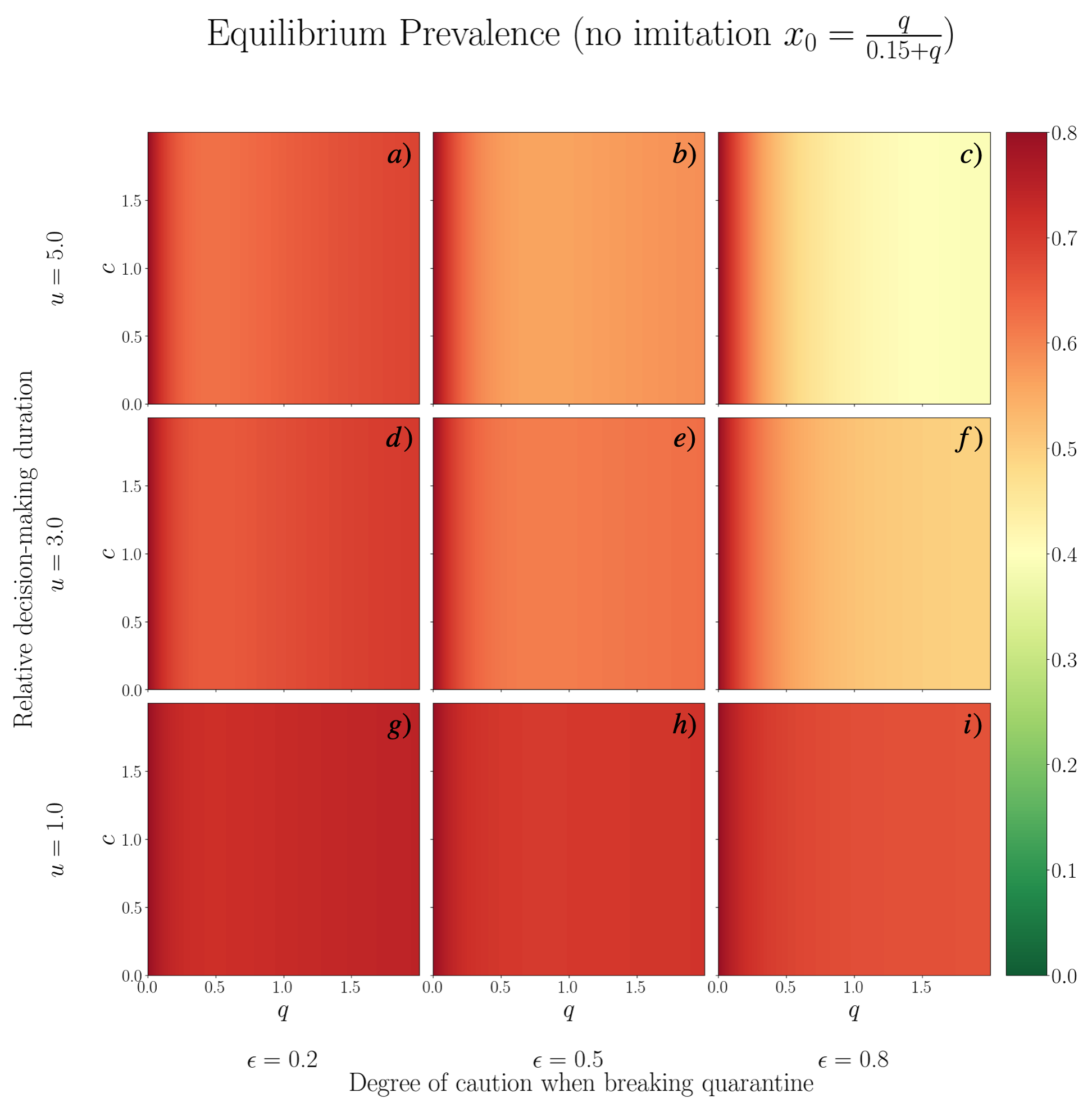}
    \caption{\textbf{Equilbrium prevalence in the absence of imitation dynamics but with compliance that decays with the duration of isolation}. Heatmap of the equilibrium prevalence depending on the cost $c$ of quarantine, the lack of coverage $q$ of the infectious period by quarantine, the rapidity of entrance in isolation $u$ and the degree of caution $\epsilon$ when breaking quarantine. We have assumed that the compliance $p_Q$ decays as $q/(0.15 + q)$, similarly to the minimal assumption made in \cite{deMeijere2021}. Parameters: $\mathcal{R}_0 = 5$ and $x_0 = 0.8$.}
    \label{fig:no_imitation_dependence}
\end{figure}

\section{Code availability}

The Python codes used throughout this paper are available at the following GitHub repository: \url{https://github.com/Hugo-Martin/SIS-quarantine}.

\end{document}